\documentclass[11pt,letterpaper]{article}
\usepackage[margin=1in]{geometry}
\usepackage[T1]{fontenc}
\usepackage{lmodern}
\usepackage{amsmath,amssymb,amsthm,mathtools}
\usepackage{comment}
\usepackage{mathrsfs}
\usepackage{microtype}
\usepackage{enumitem}
\usepackage{booktabs}
\usepackage{mathtools}
\usepackage[hidelinks]{hyperref}
\usepackage{xcolor}
\usepackage[nameinlink,capitalize,noabbrev]{cleveref}
\usepackage{graphicx}

\hypersetup{pdftitle={Approximate Polynomial Satisfiability and Integral Closure Membership in the Counting Hierarchy},
  pdfauthor={Sanyam Agarwal, Archit Hari, Anurag Pandey, Himanshu Shukla}}
\newcommand{\A}{\mathbb A}
\newcommand{\C}{\mathbb C}
\newcommand{\Q}{\mathbb Q}

\newcommand{\Z}{\mathbb Z}

\newcommand{\PPj}{\mathbb P}
\newcommand{\CH}{\mathsf{CH}}
\newcommand{\val}{\mathsf{val}}
\newcommand{\HN}{\mathsf{HN}}
\newcommand{\PSPACE}{\mathsf{PSPACE}}
\newcommand{\K}{\mathcal K}
\newcommand{\EXPSPACE}{\mathsf{EXPSPACE}}
\newcommand{\APS}{\mathsf{APS}}

\newcommand{\AV}{\operatorname{AV}}
\newcommand{\Ann}{\operatorname{Ann}}

\newcommand{\Frac}{\operatorname{Frac}}
\newcommand{\Res}{\operatorname{Res}}
\newcommand{\TT}{\operatorname{TC}}
\newcommand{\rank}{\operatorname{rank}}
\newcommand{\ord}{\operatorname{ord}}
\newcommand{\trdeg}{\operatorname{trdeg}}
\newcommand{\poly}{\operatorname{poly}}

\newtheorem{theorem}{Theorem}[section]
\newtheorem{lemma}[theorem]{Lemma}
\newtheorem{proposition}[theorem]{Proposition}
\newtheorem{corollary}[theorem]{Corollary}
\theoremstyle{definition}
\newtheorem{definition}[theorem]{Definition}
\newtheorem{example}[theorem]{Example}
\newtheorem{remark}[theorem]{Remark}
\newtheorem{problem}[theorem]{Problem}
\numberwithin{equation}{section}
\setlist[enumerate]{label=(\roman*),leftmargin=2em}
\setlist[itemize]{leftmargin=1.5em}
\title{Approximate Polynomial Satisfiability and \\ Integral Closure Membership in the \\ Counting Hierarchy}
\author{Sanyam Agarwal\thanks{Universit\"at des Saarlandes, Saarbr\"ucken, Germany.
\texttt{agarwal@cs.uni-saarland.de}}
\and Archit Hari\thanks{Indian Institute of Technology Delhi, India.
\texttt{mt1240231@iitd.ac.in}}
\and Anurag Pandey\thanks{Universit\"at des Saarlandes, Saarbr\"ucken, Germany.
\texttt{apandey@cs.uni-saarland.de}}
\and Himanshu Shukla\thanks{Universit\"at Regensburg, Germany.
\texttt{hshukla.math04@gmail.com}}}
\date{}

\begin{document}
\maketitle
\begin{abstract}
Andrews, Garg, and Schost (FOCS~2026) recently, in a breakthrough result, placed Hilbert's Nullstellensatz (HN) in the Counting Hierarchy ($\CH$) using constant-depth circuits for multivariate resultants. While HN deals with exact polynomial satisfiability, Guo, Saxena, and Sinhababu (CCC~2018; Theory of Computing~2019) introduced approximate polynomial satisfiability (APS) as a natural approximative analogue of polynomial system satisfiability. They proved its $\mathsf{NP}$-hardness and a $\PSPACE$ upper bound, and used APS to construct hitting sets for $\overline{\mathsf{VP}}$ in polynomial space, substantially narrowing the GCT chasm in geometric complexity theory.

We place APS in the Counting Hierarchy for rational polynomials given by  arithmetic circuits.  Our proof combines Noether normalization with a geometric restriction on how points in the image of a polynomial map approach the origin, reducing APS to a test on multivariate resultants 
which allows solvability in $\CH$. It complements the independent and concurrent proofs of Balaji, Shirmohammadi, Tavenas, and Worrell, and Andrews, Garg, Oliveira, and Schost. Moreover, we also place integral closure membership testing in the Counting Hierarchy. Here, the question is whether a given polynomial $f$ is integral over $\C[f_1,\ldots,f_m]$. Our approach uses a global annihilator criterion and an effective degree bound to reduce non-integrality to APS in polynomial time. In fact, we show that the reduction holds both ways, and Integral Closure Membership Testing is equivalent to APS under polynomial-time reductions.
\end{abstract}
\clearpage

\section{Introduction}\label{sec:intro}

Deciding whether a system of polynomial equations has a complex solution is a central problem in algebraic complexity. This decision problem, usually called Hilbert's Nullstellensatz ($\HN$), is $\mathsf{NP}$-hard and has long been known to lie in $\PSPACE$~\cite{Can88}. Koiran placed it in $\mathsf{AM}$ under the generalized Riemann hypothesis~\cite{Koi96}. The recent breakthrough of Andrews, Garg, and Schost gives an unconditional Counting Hierarchy upper bound~\cite{AGS26}. Their uniform constant-depth circuits for multivariate resultants extend the line of work on shallow algebraic computation developed by Andrews and Wigderson for univariate resultants and related linear-algebra problems~\cite{AW24}.

Guo, Saxena, and Sinhababu~\cite{GSS18,GSS19} introduced \emph{approximate polynomial satisfiability} (APS) as a natural approximative analogue of polynomial system satisfiability. They established algebraic and geometric characterizations of APS, proved its $\mathsf{NP}$-hardness, and gave a $\PSPACE$ algorithm. Not only that, they also showed an equivalence of non-existence of APS to the system of polynomials having an annihilator with a non-zero constant term.

Intuitively, APS allows the input coordinates to vary, and even become unbounded, provided that all output values tend to zero. For example, the system
\begin{equation}\label{example:canonicalAPS}
  f_1(x,y)=x,\qquad f_2(x,y)=xy-1
\end{equation}
has no exact solution, but the substitution $(x,y)=(t,t^{-1})$ makes its outputs $(t,0)$ tend to the origin as $t$ approaches $0$. Formally, for a polynomial map $F=(f_1,\ldots,f_m):\C^n\to\C^m$, the decision problem $\APS$ asks whether origin lies in the Zariski closure $Y:=\overline{F(\C^n)}$ of the image of the polynomial map. In other words, $\APS$ asks whether $0 \in Y$. Note that over $\C$, the Euclidean and Zariski closures of this constructible image coincide \cite[Chapter 1]{mumford1995algebraic}. 
\begin{problem}[Approximate polynomial satisfiability]\label{prob:aps}
Given  arithmetic circuits over $\Q$ computing $F=(f_1,\ldots,f_m)\in\Q[x_1,\ldots,x_n]^m$, decide whether $0\in\overline{F(\C^n)}$.
\end{problem}
Moving from an image to its closure is fundamental in border complexity: it permits algebraic objects (eg: polynomials, tensors) to be represented as limits of ``simpler'' ones. Border rank of tensors is a basic example \cite{BCS97,Landsberg12}; bounds on the order of such approximations go back to Lehmkuhl and Lickteig~\cite{LL89}. Guo, Saxena, and Sinhababu showed that APS also supplies an algorithmic route to problems in geometric complexity theory (GCT). In particular, they used it to construct polynomial-size hitting sets for $\overline{\mathsf{VP}}$ in polynomial space~\cite{GSS19}. Here $\overline{\mathsf{VP}}$ consists of polynomial families of polynomially bounded degree approximable by polynomial-size arithmetic circuits. A hitting set must detect every nonzero polynomial in this closure, including those obtained only by approximation. Through Mulmuley's connection between hitting sets and effective Noether normalization~\cite{Mul17}, the construction narrows the GCT chasm: it replaces the $\mathsf{EXPSPACE}$ bound for these constructions by a $\PSPACE$ bound. With the recent improvement in the complexity of $\HN$, it begs the question whether the complexity of $\APS$ can also be improved in the same way? Our first result gives an affirmative answer for polynomials input as arithmetic circuits with coefficients from $\Q$.

Our second result concerns integral closure, a basic construction in commutative algebra and algebraic geometry~\cite[Chapter 5]{AM69}. Given a polynomial $f$ and a finitely generated subalgebra $A\coloneqq \C[f_1,\ldots,f_m]$, we ask whether $f$ satisfies a monic equation over $A$. Formally, we ask the following.
\begin{problem}[Integral closure membership]\label{prob:int-cl}
Given  arithmetic circuits over $\Q$ computing polynomials
$f,f_1,\ldots,f_m\in\Q[x_1,\ldots,x_n]$, decide whether $f$ is integral over $A=\C[f_1,\ldots,f_m]$. In other words, we want to test whether there exists $r \in \Z_{> 0}$ and $a_0,\ldots,a_{r-1}\in A$, such that $f^r + f^{r-1}a_{r-1} + \ldots + fa_1 + a_0=0$.
\end{problem}
$\EXPSPACE$ upper bounds for the general integral closure testing are well-known through Gr\"obner Bases computations \cite{de1997algorithm}. However, we show that deciding integrality over polynomial subalgebras can be done much more efficiently than these worst-case exponential space bounds. In particular, we also put this problem in $\CH$ by reducing it to an instance of $\APS$.

\subsection{Our Contributions}

The motivation for our work comes from two related recent results: the use of homogeneous multivariate resultants by Andrews, Garg, and Schost \cite{AGS26} to put the decision problem of Hilbert's Nullstellensatz in the Counting Hierarchy, and the Strong Approximative Nullstellensatz by Agarwal, Dutta, Pandey, and Shukla \cite{ADPS26} which showed that integral closure testing over a localized subring of a polynomial algebra is in $\PSPACE$. Knowing $\HN \in \CH$ naturally leads to the question whether we can improve the $\PSPACE$ upper bound for $\APS$ in the same way, also posed as an open problem in \cite{ADPS26}. Our first main result answers this question positively. 
\begin{theorem}\label{thm:aps}
Approximate polynomial satisfiability ($\APS$) belongs to the Counting Hierarchy. 
\end{theorem}
Input polynomials are represented by arithmetic circuits over $\Q$, with binary rational constants. Division gates are allowed, under the promise that every divisor is nonzero as a rational function and that every output represents a polynomial. Complexity statements use the ordinary Turing model and binary encodings of rational coefficients.

\begin{remark}[Sparse representation]
Our results also apply to polynomials given in sparse representation, with exponents encoded in binary. Indeed, repeated squaring converts such inputs into division-free arithmetic circuits in time polynomial in the sparse input size.
\end{remark}
\begin{remark}[Division Gates]
    Note that the proof also holds if circuits of input polynomials $f_i$ have division gates. If degree of each $f_i$ is polynomial in input length, then we can use \cite{strassen1973vermeidung} to remove these division gates with only a polynomial blowup. On the other hand, even if the degree of $f_i$ were exponential in the input description, then writing each $f_i = \frac{P_i}{Q}$, it is easy to see that the system $F = (f_1(x),\ldots,f_m(x))$ has an APS if and only if the system $G=(uP_1(x),\ldots,uP_m(x), uQ(x)-1)$ has an APS, where $u$ is a new variable.  
\end{remark}
\begin{remark}[Coefficients from other fields]
    Theorem~\ref{thm:aps} also holds for arithmetic circuits over number fields and finite fields, with image closure taken over the algebraic closure of the coefficient field, see Remark~\ref{rem:APS-for-other-fields}. For simplicity of exposition, throughout this work, we assume that the input polynomials have rational coefficients.
\end{remark}

Several standard applications of $\APS$ consequently inherit
improved upper bounds. Below, $k$ is any of the coefficient
fields covered above, $K=\overline{k}$, and vectors and tensors are given by their full coordinate lists.

\begin{corollary}\label{cor:aps-applications}
Theorem~\ref{thm:aps} has the following consequences.
\begin{enumerate}
    \item \textbf{Tensor border rank.}
    Given a tensor over $k$ and an integer $r$ encoded in binary, deciding whether its border rank over $K$ is at most $r$ belongs to $\CH$.

  \item \textbf{Orbit-closure problems.} For rational linear actions of products of general and special linear groups, given by arithmetic circuits in explicit matrix coordinates, deciding
\[
    w\in\overline{G\cdot v},
    \qquad\text{equivalently,}\qquad
    \overline{G\cdot w}\subseteq\overline{G\cdot v},
\]
belongs to $\CH$. This includes null-cone membership, obtained by taking $w=0$.
    \item \textbf{Hitting sets for $\overline{\mathsf{VP}}$.} Given $n,s,d$ in unary, a hitting set of size $\operatorname{poly}(n,s,d)$ for all nonzero $n$-variate polynomials over $K$ of degree at most $d$ and approximative circuit size at most $s$ can be constructed in $\mathsf{FP}^{\CH}$. Thus these hitting sets are $\CH$-explicit. Their coordinates have polynomial-length encoding; over finite fields, we allow a finite extension of polynomial degree.
\end{enumerate}
\end{corollary}

The first two consequences follow from the standard parametrizations of bounded-rank tensors and group orbits, using denominator clearing for rational parametrizations. For hitting sets, the certification criterion of Guo, Saxena, and Sinhababu~\cite[Theorem~5.2 and Remark~5.3]{GSS19} reduces failure of a candidate hitting set for their universal circuit family to $\APS$. Certification therefore belongs to $\CH$, and bitwise search over polynomial-length candidate encodings gives the stated construction, improving their $\PSPACE$ upper bound.

For our second main result, we provide a solution to the question raised in Problem~\ref{prob:int-cl}, bringing down the complexity of integral closure membership testing from $\mathsf{EXPSPACE}$ to $\CH$.

\begin{theorem}\label{thm:icm}
Integral closure membership testing belongs to $\CH$. More precisely, non-integrality has a deterministic polynomial-time many-one reduction to $\APS$.
\end{theorem}

\begin{remark}[Other coefficient fields]
Theorem~\ref{thm:icm} also holds over the number fields and finite fields covered by our extension of Theorem~\ref{thm:aps}, with integrality taken over $\overline{k}[f_1,\ldots,f_m]$, where $k$ is the base field, see Remark~\ref{rem:icm-other-fields}.
\end{remark}

\begin{corollary}\label{cor:integrality-applications}
Under the input and coefficient-field conventions above, the following problems belong to $\CH$, where $K$ denotes the algebraic closure of the coefficient field:
\begin{enumerate}
    \item \textbf{Integral closure of ideals.} Given polynomials $h,g_1,\ldots,g_m$, decide whether $h$ belongs to the integral closure of the ideal $(g_1,\ldots,g_m)\subseteq K[x_1,\ldots,x_n]$.

    \item \textbf{Finiteness of polynomial maps.}     Given a polynomial map $F=(f_1,\ldots,f_m):K^n\to K^m$, decide whether $F$ is finite, or equivalently, proper.
\end{enumerate}
\end{corollary}

The first consequence follows from the standard equivalence between integrality of $h$ over the ideal ${(g_1,\ldots,g_m)}$ and the integrality of $th$ over $K[x_1,\ldots,x_n,tg_1,\ldots,tg_m]$, where $t$ is a new variable. To see this consider the degree 1 graded part, i.e., $k=1$ 
in the proof of ~\cite[Proposition~5.2.1]{HS06} with $R=K[x_1,\ldots, x_n]$, and $I=(g_1,\ldots, g_m)$. For the second, $F$ is finite precisely when every coordinate $x_j$ is integral over $K[f_1,\ldots,f_m]$; finiteness and properness are equivalent for affine morphisms \cite[Tag~01WN]{Stacks26}.

Moreover, we also show that Integral Closure Membership Testing is equivalent to APS under polynomial-time reductions. \begin{theorem}\label{thm:aps-icm-equal} 
    Under polynomial time many-one reductions, non-existence of $\APS$ is equivalent to Integral Closure Membership Testing.
\end{theorem}

\begin{remark}
    The above equivalence works over all algebraically closed fields.
\end{remark}

\subsection{Related work}

\paragraph{Independent and concurrent work.}
Independently and concurrently, Balaji, Shirmohammadi, Tavenas, and Worrell~\cite{BSTW26} obtained the Counting Hierarchy upper bound for $\APS$ established in \Cref{thm:aps}. Both proofs use generalized characteristic polynomials and the resultant circuits of Andrews, Garg, and Schost~\cite{AGS26}, but their geometric reductions differ. \cite{BSTW26} first reduce to a minimally dependent tuple and prescribe generic scalar multiples of a single infinitesimal parameter as the values of all but one coordinate. Their key argument uses the lowest-degree homogeneous part of an annihilating relation to show that, for a suitable choice of the remaining coordinate, a positive APS instance admits a solution on which that coordinate also tends to zero. They detect the existence of such a solution by comparing resultant valuations.

Our geometric reduction builds on ideas of Lehmkuhl and Lickteig~\cite{LL89} from their study of tensor border rank. We use finite Noether normalization of the image closure, reducing origin membership to membership in the finite fiber over zero. We then show that every point of this fiber is approached by image points whose projections lie on one generic line through zero. Crucially, finiteness of the normalization also ensures that the image coordinates of every solution branch have finite limits. Thus, the possible limits are exactly the points of this finite fiber. Applying a generic linear form allows us to represent these points by the roots of a univariate polynomial, with zero occurring as a root exactly when the original instance admits an approximate solution. We obtain this polynomial through the resultant construction of~\cite[Section~4]{AGS26}.

We also learned through personal communication that Andrews, Garg, Oliveira, and Schost~\cite{AGOS26} independently obtained the same upper bound. In concurrent work, Bhargav, Curticapean, Jindal, and Shukla~\cite{ETRinCH26a} establish
$\#\mathsf{P}$-hardness of computing a bit of the multivariate resultant.

\paragraph{Strong Approximative Nullstellensatz and its connection to integrality.}
Guo, Saxena, and Sinhababu characterize non-existence of APS by an annihilator with non-zero constant term and use bounded-degree relations and succinct linear algebra to obtain their $\PSPACE$ bound~\cite{GSS19}. This is very similar in spirit to the Weak Hilbert's Nullstellensatz. Inspired by this, Agarwal, Dutta, Pandey, and Shukla develop a strong approximative Nullstellensatz for containment of approximate vanishing sets~\cite{ADPS26}. Their criterion involves integral dependence over the maximal ideal of a \emph{local} subalgebra of $\C(f_1,\ldots, f_m)$, thus, allowing the use of elements from $\C(f_1,\ldots,f_m)$. While similar in nature, the problem is distinct from checking integrality over $\C[f_1,\ldots,f_m]$. For example, given $f=x^2y$ and $g=xy-1$, it can be checked that $h=x$ is not integral over $\C[f,g]$, but $h = \frac{f}{1+g}$, and hence lies in the localized subalgebra used in \cite{ADPS26}.

\paragraph{Polynomial Subalgebra Membership Testing.}
Kayser showed that Subalgebra Membership Testing is $\mathsf{EXPSPACE}$-complete, and that its homogeneous restriction is $\PSPACE$-complete~\cite{Kay25}. However similar, subalgebra membership is not the same as integral closure. For example, $x$ is integral over $\C[x^2]$ but does not belong to it.

\paragraph{Resultants, $\mathsf{ETR}$, and the Counting Hierarchy.}
Andrews and Wigderson gave constant-depth arithmetic algorithms for univariate resultants, greatest common divisors, and related problems~\cite{AW24}. This was later extended by Andrews, Garg, and Schost to multivariate resultants and they used it to place HN in the Counting Hierarchy~\cite{AGS26}, using generalized characteristic polynomials originating with Ierardi and Canny~\cite{Ier89,Can90} to treat affine systems. Recently, \cite{ETRinCH26a, ETRinCH26b} have shown that the complexity class Existential Theory of Reals ($\mathsf{ETR}$) also lies in the Counting Hierarchy, using the method of multivariate resultants. Additionally, \cite{ETRinCH26a} also show that computing a bit of this multivariate resultant is $\#\mathsf{P}$-Hard. 

\paragraph{Annihilator degree bounds.}
Degree bounds for annihilating polynomials were first described by Perron~\cite{Perron27}, and this line of work was further extended in \cite{Kay09, BeeckenMittmannSaxena13} that showed singly exponential bounds for the minimal annihilator. In particular, these bounds were used by Guo, Saxena, and Sinhababu~\cite{GSS19} to obtain an annihilator of singly exponential degree, which they then use to decide non-existence of $\APS$ in $\PSPACE$. In the same vein, we also use these bounds to establish the singly exponential degree bound for our integral closure membership testing certificate. 

\section{Proof Ideas}
We first describe our methods for proving \cref{thm:aps}.
\subsection{Approximate Polynomial Satisfiability in Counting Hierarchy} \label{sec:aps_overview}

Let
\[
  F=(f_1,\ldots,f_m):\mathbb C^n\longrightarrow\mathbb C^m,
  \qquad
  Y=\overline{F(\mathbb C^n)}.
\]
We want to decide whether \(0\in Y\). Write \(r=\dim Y\), and first assume \(r\geq1\). Without loss of generality, we can assume that $m=r+1$ and $n=r$ using the standard tricks which preserve existence of APS such as the random reduction theorem~\cite{GSS19} and the random hyperplane intersection lemma~\cite{AGS26}, see Lemma \ref{lem:rewrite-random-sections} and \cref{thm:rewrite-random-projection} for precise statements. So for the rest of this section we assume $n=\dim(Y)$, and $m=n+1$. 

\begin{figure}[t]
  \centering
  \includegraphics[width=\linewidth]{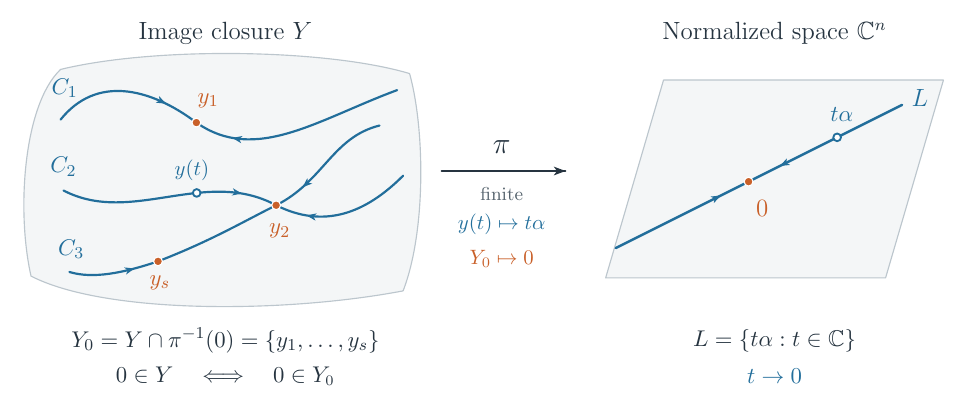}
  \caption{A finite Noether normalization $\pi:Y\to\mathbb \C^n$ reduces testing $0\in Y$ to testing membership in the finite set $Y_0=Y\cap\pi^{-1}(0)$ (orange points). For a generic line $L=\{t\alpha:t\in \C\}$ through the origin, the irreducible components $C_i$ of $Y\cap\pi^{-1}(L)$ are curves mapping onto $L$. Every point of $Y_0$ lies on at least one such curve. It can be approached along a branch with $\pi(y(t))=t\alpha$, using actual image points away from the limit. The parameter may require a finite algebraic extension of $\C(t)$. Arrows indicate specialization as $t\to0$. 
  }
  \label{fig:normalization-fiber}
\end{figure}
On a very high level, our proof has these main steps: (a) Initially, we need to check $0 \in Y = \overline{F(\C^n)}$, where $Y$ might be an infinite set. However, we show that we can restrict ourselves to checking membership of origin in a \emph{finite} subset of $Y$. We do this using the Noether  normalization which gives a linear map $\pi$ with the property that only a finite subset $Y_0 \subseteq Y$ maps to $0$ in the normalized space, that is, $\pi(Y_0) = \{0\}$. (b) We then show that every point \(y_0 \in Y_0\) can be approached by those points in the image $F(\C^n)$ which project to the same generic line $L$ in the normalizing space under the map $\pi$ (see \cref{fig:normalization-fiber} for a visualization). Parameterizing the line $L$ by $t$, we get that all solutions $x(t)$ of $(\pi\circ F)x(t) = t \alpha $ (where $\alpha$ is the direction of line $L$) satisfies $\lim_{t \to 0} F(x(t))=y_0$ for some $y_0 \in  Y_0$. Thus, our task of checking if origin is in the closure of the image is reduced to checking if $\lim_{t \to 0} F(x(t))=0 $ for some $x(t)$ with $(\pi\circ F)x(t) = t \alpha $. Moreover, Noether normalization and fiber-dimension theorem guarantees that such $x(t)$ will only be finitely many. (c) Next, we show that in order to check if $\lim_{t \to 0} F(x(t))=0 $ for some solution $x(t)$, it suffices to check if  $\lim_{t \rightarrow 0} (\Sigma_i \lambda_i f_i(x(t))) = 0$ for some solution $x(t)$, for a generic choice of $\lambda \in \C^m$. (d) We then define a univariate polynomial with roots $\Sigma_i \lambda_i f_i(x(t))$ corresponding to the solutions $x(t)$, which has $0$ as a limiting root $\Sigma_i \lambda_i f_i(x(t))$ if and only if $0 \in Y$. (e) This allows us to use the multivariate resultant method of \cite{AGS26} - with slight modifications - to check if this affine root indeed exists. 

\vspace{1em}

We now describe each of these steps in more details. 
\vspace{-1em}
\paragraph{1. Restricting to a finite subset of the image closure.}
Linear Noether normalization, is a standard tool in algebraic geometry which gives a linear map $\pi: \C^m \rightarrow \C^n$ such that
\[
  \pi(y)=
  \left(
    \sum_{i=1}^m a_{1i}y_i,\ldots,
    \sum_{i=1}^m a_{ni}y_i
  \right)
\]
whose restriction to \(Y\) is finite and surjective, so $\pi|_Y$ takes Zariski closed sets to closed sets, see Lemma~\ref{lem:rewrite-linear-nn} for a precise statement. In particular, $Y_0=\{y\in Y:\pi(y)=0\}$ is finite. Since \(\pi(0)=0\),
\[
  0\in Y\quad\Longleftrightarrow\quad 0\in Y_0.
\]

\paragraph{2. Constructing a parametrized system of equations with finitely many solutions.}
Note that some of these points in $Y_0$ will originate from $\C^n$ under the polynomial map $F$, while the others will be points purely in $Y \setminus F(\C^n)$, and hence have no preimage. If $0$ was indeed in $F(\C^n)$, then this corresponds to an actual root of the system, and the problem reduces to Hilbert's Nullstellensatz, which is known to be in $\CH$~\cite{AGS26}. Hence, we assume that $0$ has no preimage. 

We then show that every point \(y_0 \in Y_0\) can be approached by those points in the image $F(\C^n)$ which project to the same generic line $L\subset \C^n$ in the normalizing space under the map $\pi$. For this, choose a generic direction \(\alpha\in\mathbb C^n\) corresponding to a line $L\subset \C^n$ through the origin. Identify the one dimensional space $\C$ with $L$ via $t\mapsto \alpha t$. Then, $\pi^{-1}(L)$ can be split into irreducible curves, each of which is surjective on $L$, see Lemma~\ref{lem:horizontal}. Since $0 \in L$, any point of $Y_0$ necessarily lies on one of these curves, see \cref{fig:normalization-fiber}.

By another standard consequence of linear Noether normalization (Lemma~\ref{lem:rewrite-linear-nn}) the coordinates $y_i$ of $Y$ are integral elements over $\C[u_1,\ldots, u_n]$,
where $\pi(y)=(u_1,\ldots, u_n)$, i.e., $y_i$ satisfy monic polynomials equations over $\C[u_1,\ldots, u_n]$. Replacing $u_i$ by $\alpha_it$ makes $y_i(t)$ integral over $\C[t]$. Hence, no $y_i(t)$ can ``diverge'' at $t=0$. In other words, the valuation of each $y_i(t) \geq 0$, see Section~\ref{sec:prelim-aps-geometry}.

The connection to the approximate solutions can now be viewed as follows. Since the image of $F$ and its domain have the same dimension, the fiber-dimension theorem (see \cref{thm:dominance-function-fields}) gives a Zariski open set $U\subset Y$
such that for each $y\in U$,
the fiber set $\{x\in\C^n\ :\  F(x)=y\}$ is finite.
Since $\pi$ is surjective, there is a Zariski open set $V\subset\C^n$ such that $V = \A^n \setminus \pi(Y \setminus U)$, and $\pi^{-1}(V) \cap Y \subseteq U$. A general direction $\alpha$ will lie in $V$ with high probability. Hence, for such an $\alpha\in V$, the line $L$ is ``almost'' contained in $V$, i.e., $L\cap V$ misses only finitely many points of $L$. Thus every point \(y_0 \in Y_0\) can be approached by actual outputs \(F(x(t))\) satisfying
\begin{equation}\label{eq:idea-line}
  \sum_{i=1}^m a_{ji}f_i(x(t))=t\alpha_j,
  \qquad j=1,\ldots,n,
\end{equation}
along an algebraic branch $x(t)\in \overline{\C(t)}^n$ such that 
$F(x(t))\in U$.
Hence, for every \(y_0\in Y_0\), there is an algebraic branch \(x(t)\) with
\[
  F(x(t))\longrightarrow y_0,
  \qquad
  \pi(F(x(t)))=t\alpha.
\]

Note that the output paths \(F(x(t))\) may differ for different points, and only their projections lie on the fixed line \(t\alpha\). However, this allows us to search for the desired ``limiting output'' using the exact equations \eqref{eq:idea-line}. Concretely, we work with the following system of polynomials:
\begin{equation}\label{eq:idea-finite}
  \begin{cases}
    f_i(x)-y_i=0
      & 1\leq i\leq m=n+1,\\
    \displaystyle\sum_{i=1}^m a_{ji}y_i-t\alpha_j=0
      & 1\leq j\leq n.
  \end{cases}
\end{equation}
By the above discussion, this system has a finite nonempty set of solutions, and we have $2n+1$ equations in $2n+1$ variables $x_1,\ldots,x_n,y_1,\ldots,y_{m}$.
Furthermore, every solution of \eqref{eq:idea-finite} has a finite limit in the $y$-coordinates, ($y$ coordinates are integral over $\C[t]$ as discussed) and that limit belongs to $Y_0$. Since $Y_0$ is covered exactly by these limits, the possible output limits are \emph{exactly} the points of \(Y_0\). Thus, we can restrict our attention to identifying if $0$ occurs as a limit, see Section~\ref{sec:geometry-rewrite} for a detailed treatment. In particular, notice that we need to look at the algebraic solutions using Puiseux
series and take the projective closure of a curve defined by Equation \eqref{eq:idea-finite} on the graph associated to the polynomial system $F$ to ensure that the ideas we describe extend to homogeneous systems required by the multivariate resultant.

\paragraph{3. Constructing the univariate polynomial associated to $Y_0$.}
Rename the $x$ coordinates to $z_1,\ldots,z_n$ and $y$-coordinates 
as $z_{n+1},\ldots, z_{n+m}$ in 
Equation~\eqref{eq:idea-finite}.
  Since $Y_0$ is finite, we can always choose \(\lambda\in\mathbb C^m\) such that
\(\lambda\cdot y\ne0\) for every \emph{non-zero} \(y  \in Y_0\), and define
\[
  \ell(z)=\sum_{i=1}^m\lambda_i z_{n+i}
      =\sum_{i=1}^m\lambda_i f_i(z_1,\ldots, z_n).
\]
This gives us the following equivalence:
\begin{equation}\label{eq:0inY-lim_h(z)is0}
    0\in Y
  \quad\Longleftrightarrow\quad
  \text{some solution of \eqref{eq:idea-finite} satisfies }
  \lim_{t\to0}\ell(z(t))=0.
\end{equation}
 Using the above defined 
linear form $\ell$, construct the 
polynomial
\[Q_{Y_0,\ell}(T)\coloneqq \prod\limits_{y\in Y_0}(T-\ell(y)).\]
So $0\in Y_0$ exactly when $Q_{Y_0,\ell}(0)=0$.
Recall that each $y\in Y_0$ 
is obtained as a limit of some 
solution $z_y(t)$ of \eqref{eq:idea-finite}, and each 
solution $z(t)$ of
\eqref{eq:idea-finite} approaches a point in $Y_0$ in $z_{n+i}$ coordinates.
So the polynomial 
\[P_\ell(T,t)=\prod
\limits_{z(t)\text{ a solution to \eqref{eq:idea-finite}}}(T-\ell(z(t)))\]
has the same zeros at $t=0$ as $Q_{\ell,Y_0}(T)$.

The observation that
$P_\ell(T,t)$ is
obtained as a product of evaluations of a linear polynomial on the solutions of a 
polynomial system \eqref{eq:idea-finite} intuitively prompts us towards using resultants. Precisely,
the \emph{Poisson formula} \cite[Theorem 4.3]{AGS26} allows us to evaluate the new polynomial on roots of the rest of the polynomials.

\paragraph{4. Using the multivariate resultant.} 
Let $Z$ be the solutions to the system \eqref{eq:idea-finite} over $\overline{\C(t)}$. Note that this system need not be homogeneous. On the other hand, the multivariate resultant is only defined for homogeneous systems as discussed in \cite{AGS26}. Homogenizing the system \eqref{eq:idea-finite} considered as polynomials over $\C(t)$ can introduce positive dimensional components at infinity that are not closures of their affine solutions. In \cite{AGS26} the authors tackle exactly this issue by using the generalized characteristic polynomials introduced by Canny-Ierardi \cite{Can90,Ier89} which uses a perturbation of these polynomials. In general, let $G_1,\ldots,G_N$ be the homogenizations of \eqref{eq:idea-finite}, with 
$N=2n+1$ and $d_i$ is the degree of $G_i$ in $z_0,\ldots, z_N$ variables over $\C(t)$. Note that we are only homogenizing with respect to the $z$ variables. Introduce $\varepsilon$ and replace them by $\widehat G_i=G_i+\varepsilon z_i^{d_i}$. 
This removes the positive dimensional component but some finite set of points at infinity might still survive the Canny-Ierardi perturbation, for example $F=(x_1+x_2,x_1+x_2-1)$ has no APS, but the homogenized system $F=(x_1+x_2,x_1+x_2-x_0)$ has a root $[0:1:-1]$.
Since we are concerned with the solutions to the affine system \eqref{eq:idea-finite} over $\C(t)$,
we introduce a new parameter $w$, that will track the points at infinity. So similarly to \cite{AGS26} define
\[
 H_{T,w}=T z_0^2-z_0\ell(z)+\sum_{j=1}^Nw^jz_j^2,
\]
be a quadratic form over $\C(T,w)$. 
At an affine point $z$
with $z_0=1$, and $w=0$ $H_{T,0}(z)=T-\ell(z)$,
which is exactly the form of the polynomial $P_\ell$.
So taking the multivariate resultant of $G_i$ and $H_{w,T}$ should be connected to $P_\ell$, and so to $Q_{Y_0,\ell}$.
Let
\begin{equation}\label{eq:overview-resultant}
 \mathscr R(\varepsilon,t,w,T)=
 \Res_{z_0,\ldots,z_N}
       (\widehat G_1,\ldots,\widehat G_N,H_{T,w}).
\end{equation} 
 We now use the finite limit set and factorization statements in~\cite[Lemmas~4.6--4.7]{AGS26}. The limit set contains every isolated point of the original projective zero set, so every point of \(Z\). Its affine points are precisely the set $Z$. Thus, we get \[
  \widetilde{P_\ell}(T,t):=
  \operatorname{TC}_w\operatorname{TC}_{\varepsilon}\mathscr R
  =
  a(t)\prod_{z\in Z}
            (T-\ell(z(t)))^{\mu_z},
  \qquad a(t)\in\C[t],\quad\mu_z>0.
\]
Here $\TT_v(\cdot)$ is the 
least degree coefficient of 
a polynomial in variable $v$.
The polynomial $P_\ell$ and 
$\widetilde{P_\ell}$ differ
by a non-zero factor $a(t)$.
Thus, writing $a(t) = t^k v(t)$ such that $v(0)\neq 0$, we get  $R_\ell(T)\coloneqq \TT_t \widetilde{P_\ell}(T,t)=cP_\ell(T,0)$ for some \(c\in\mathbb C^\times\). Hence, by using multivariate resultants we construct a polyomial $R_\ell$ that has same zeros as $Q_{Y_0,\ell}$, see Section~\ref{sec:rewrite-resultant}, and Section~\ref{sec:same-root} for details.

\paragraph{5. Perform the resulting coefficient test in CH.}
The system of homogeneous polynomials $\widehat G_i, H_{T,w}$ has polynomial size. By \cite{AGS26} we can access any specified coefficient bit of its resultant in $\CH$, hence we do not need to write out the entire resultant.  

After clearing denominators, the resultant has degree and integer coefficient bit lengths bounded by $2^{\poly(L)}$, where $L$ is the input length. Thus both monomial exponents and positions within a coefficient's binary representation have polynomial-length encodings. Given $\CH$ access to individual coefficient bits, coefficient vanishing can be tested by quantifying over these bit positions. A fixed number of quantifier blocks over polynomial-length indices then suffices to identify the successive lowest nonzero powers and test the resulting constant coefficient in $T$.
Hence the entire test lies in $\CH$; see Section~\ref{sec:effective}.

\subsection{Integral Closure Membership Testing in Counting Hierarchy}
In this section, we discuss the ideas we use to place the problem of deciding  whether a polynomial $f$ lies in the integral closure of the ring $A=\C[f_1,\ldots,f_m]$ in $\CH$ (see Problem \ref{prob:int-cl}), and the ideas needed to establish the computational equivalence of $\APS$ and Integral Closure Membership Testing. 

In order to prove this we want to use the above $\APS \in \CH$ result. We do precisely that. We give a polynomial time reduction from integral closure membership testing to $\APS$ to get the desired result. Apriori, it is not clear how integrality relates to $\APS$. However, we observe that the Approximate Rabinowitsch's Trick~\cite[Lemma 2]{ADPS26} exactly hints towards this connection. In fact, we can use the classical \emph{valuative criterion} for integral dependence \cite[Corollary 5.22]{AM69} (see Corollary~\ref{cor:boundedness} for a precise statement) to make this connection precise. In particular, we can define a system of polynomials $\Psi_M = (uf_1^M,\ldots, uf_m^M, u, 1-uf)$ such that $f$ is integral over $\C[f_1,\ldots,f_m]$ if and only if $\Psi_M$ has no
APS, where $M$ is singly exponential in the degrees of the input polynomials. See Section~\ref{sec:integral} for a full proof. This also establishes that Integral Closure Membership Testing reduces to non-existence of $\APS$. To show the reduction in the other direction, we observe that $f_1,\ldots,f_m \in \C[x_1,\ldots,x_n]$ have no APS if and only if $z$ is integral over the polynomials $zf_1,\ldots,zf_m \in \C[z,x_1,\ldots,x_n]$, see Lemma~\ref{lem:aps-to-nonintegrality}.

We now describe the proof idea for \cref{thm:icm} in more details. For the proof we need some concepts from valuation theory of rings \cite{AM69}. In particular, for $\beta_i \in \C(t)$, we say that the valuation of a polynomial $g \in \C[x_1,...,x_n]$ (denoted $\val_t(g(x))$) is the exponent of the least power of $t$ in $g(\beta_1(t),\ldots,\beta_n(t))$. If $g(\beta)=0$, then valuation is defined to be infinity. For a precise definition, see Section~\ref{sec:prelim-approximation}. \begin{example}[Valuations]\label{example:valuations}
    Let $g = x_1x_2-1$. \begin{enumerate}
        \item If $\beta_1 = \beta_2 = t$, then $g(\beta) = t^2-1$, and $\val_t(g(\beta)) = 0$.
        \item If $\beta_1 = t; \beta_2 = t + \frac{1}{t}$, then $g(\beta) = t^2$, and $\val_t(g(\beta)) = 2$.
        \item If $\beta_1 = t; \beta_2 = \frac{1}{t}$, then $g(\beta) = 0$, and $\val_t(g(\beta)) = \infty$.
    \end{enumerate}
\end{example}

By definition, $f$ is integral if and only if there is a monic polynomial in $f$ with coefficients from $A$. Observe that this monic polynomial guarantees that if valuation of each $f_i$ is non-negative then valuation of $f$ will also be non-negative. At this point, it is not so clear if the converse also holds. But the {valuative criterion} for integral dependence says exactly this, see Corollary~\ref{cor:boundedness}. In other words, it means that if $f$ is not integral, then there is an assignment $\beta$ such that all the $f_i$ are bounded (i.e., $\val_t(f_i(\beta)) \geq 0$), but $f$ diverges (i.e., $\val_t(f(\beta)) < 0$). We observe that this is very similar to the approximate Rabinowitsch's trick relating to non-divergence used in \cite[Lemma 2]{ADPS26} which gives a criterion for this non-divergence in terms of APS. Since we have shown that $\APS$ is in $\CH$, this gives an approach towards also putting integral closure membership testing in $\CH$. However, at this point there is still a computational roadblock as this requires deciding APS for $f_i-c_i$, for each $c_i\in \C$. 
Precisely, it requires deciding whether for each $c \in \C^m$, $f_1-c_1,\ldots,f_m-c_m,u,1-uf$ have no APS. 

In terms of annihilators, using \cite{GSS19} criteria, the system $f_1-c_1,\ldots, f_m-c_m, u, 1-uf$ might give different annihilators for different $c\in \C^m$. We show in \Cref{thm:pointwise-global-equivalence} that
in fact the ``local'' annihilators for each $c\in \C^m$ can be used to get a single \emph{global} polynomial certificate $Q$ of degree $D$ which is singly exponential in the input parameters, see~\cref{thm:exp-deg-bd-annihilator} for a full statement. At this point, since linear algebra over a singly exponential system is in $\PSPACE$~\cite{csanky1976fast,Ber84}, we can check for integrality in $\PSPACE$. Still, we cannot reduce the complexity to $\CH$ yet, as this global polynomial doesn't directly correspond to an $\APS$ instance.

To achieve this, we again use Corollary~\ref{cor:boundedness}. We claim that setting $M>D$, we can define a system of polynomials $\Psi_M = (uf_1^M,\ldots, uf_m^M, u, 1-uf)$ such that $f$ is integral over $\C[f_1,\ldots,f_m]$ if and only if $\Psi_M$ has no
APS, see \cref{thm:power-reduction} for a full proof. For the easy direction, $f$ 
not being integral implies
that there is an algebraic branch $x(t)$ such that each $f_i$ had a non-negative valuation but $f$ had a negative valuation. Then, select $u = \frac{1}{f}$, and so $u$ has a positive valuation. Thus each of the polynomials in $\Psi_M$ has a positive valuation, and hence there is an APS. The converse is where the choice of $M$ becomes crucial. Intuitively, it ensures that the valuation of $u$ doesn't dominate the valuation of any $f_i$ and we don't get false witnesses of non-integrality. 
\begin{example}\label{example:integrality-witness}
   Let $f_1 = x$ and $f = x^2$. Clearly, $f$ is integral over $\C[f_1]$, yet if we set $M=1$, $(u,1-ux^2,ux)$ has the approximate zero $(x,u)=(t^{-1},t^2)$.
   Here, the generator $f_1$ is also diverging, just more slowly than $f$ along $x=1/t$.
\end{example}
 Suppose $f$ is integral and there is an APS of $\Psi_M$, and let $x(t), u(t)$ be an approximate solution witnessing APS. 
 Thus, $a=\val_t(u(t))$ and each $u(t)f_i^M(x(t))$ is 
 positive. Hence, valuation of each $f_i(x(t))$ is greater than $-\frac{a}{M}$. Now using the global degree $D$ certificate $Q$ obtained above we get the identity
\[
  1+\sum_{j=1}^{r}b_j(f_1,\ldots,f_m)f^{-j}=0.
\]
At an approximate zero of $\Psi_M$, the choice $M>D$
makes every summand after $1$ have positive valuation. They therefore cannot cancel the constant term, giving a contradiction.
 
 Since $M$ is singly exponential, we can compute each $uf_i^M$ in polynomial depth and get a polynomial sized APS criterion for the integral closure test putting integral closure testing in $\CH$.

\section{Preliminaries}\label{sec:prelim}
Let $K$ be an algebraically closed field.
We write $x=(x_1,\ldots,x_n)\in\A^n_K$, for a general
point in the affine $n$-space $\A^n_K$ over $K$, 
and also for coordinates on $\A^n_K$. The distinction 
will be clear from the context. We use dimension in the sense
of~\cite[Definition~12.1.1]{Vak25}. The degree of an affine variety
is the degree of its projective closure, with projective degree defined
in~\cite[Section~18.6, p.~520]{Vak25} 

\subsection{Dominant and finite morphisms, and integral dependence}
\label{sec:prelim-integrality}
Throughout this paper, an 
affine algebraic set will be the same as an algebraic variety, and when an algebraic variety it irreducible we will explicitly state this. 
We first recall the definition of coordinate ring associated to 
an affine algebraic set over $K$. 
Let $X\subset \A^n_K\simeq K^n$ (the last isomorphism
is as a set) be an algebraic set, and 
$I_X\subset K[x_1,\ldots,x_n]$ be
its assoicated ideal, then the cooridinate ring is given
by $K[X]\simeq K[x_1,\ldots, x_n]/I_X$. 
A morphism of affine algebraic sets $f:X\to Y$ given 
by polynomials induces a ring homomorphism 
of associated coordinate rings $f^*:K[Y]\to K[X]$
in the opposite direction.

\begin{definition}[{\cite[Section~3.6.4]{Vak25}}]
\label{def:irreducible-affine-variety}
Let $K$ be an algebraically closed field. An \emph{irreducible affine
variety over $k$} is a nonempty Zariski-closed subset
$X\subseteq\A^n_K$ that cannot be written as the union of two proper
Zariski-closed subsets of $X$. Equivalently, the associated ideal $I_X$ is a prime 
ideal.
\end{definition}
For an irreducible affine variety $X$ over $K$, since
$I_X$ is a prime ideal its coordinate ring $K[X]$ is 
an integral domain. Hence, 
its \emph{function field} is $K(X)=\Frac(K[X])$, where $\Frac(K[X])$ is the field
of fractions of $K[X]$.
\begin{example}\label{ex:example-irred}
  Any affine space $\A^n_K$ is by definition 
  irreducible. If $X$ and $Y$ are irreducible, then 
  so is $X\times Y$. 
  Furthremore, if $f:X\to Y$ is map  affine
  algebraic sets, and $X$ is irreducible, then the 
  image $f(X)$ is irreducible. In particular, for
  $X=\A^n_K$, $f(\A^n_K)$ is always irreducible. 
\end{example}
\begin{definition}[{\cite[Tag~\href{https://stacks.math.columbia.edu/tag/01RI}{01RI}]{Stacks26}
}]
\label{def:dominant-morphism}
Let $X$ and $Y$ be irreducible affine varieties over an algebraically
closed field $K$. A morphism $f:X\to Y$ is \emph{dominant} if its
image is Zariski dense in $Y$, that is,
\[
 \overline{f(X)}=Y,
\]
where the closure is taken in $Y$. Equivalently,
$f^{-1}(U)\ne\varnothing$ for every nonempty Zariski-open subset
$U\subseteq Y$.
\end{definition}

We recall the following properties of dominant
morphisms~\cite[Exercise~7.5.C and Theorem~12.2.1]{Vak25}.

\begin{theorem}[{\cite[Theorem~8.4.2,Corollary 12.4.2]{Vak25}}]
\label{thm:dominance-function-fields}
Let $X$ and $Y$ be irreducible affine varieties over an algebraically
closed field $K$, and let $f:X\to Y$ be a dominant morphism.
Then $f$ induces a field embedding $f^*:K(Y)\hookrightarrow K(X)$.
Identifying $K(Y)$ with its image, we have
\[
 \trdeg_{K(Y)}K(X)=\dim X-\dim Y.
\]
Moreover, there is a nonempty Zariski-open subset $U\subseteq Y$ such
that, for every $y\in U$, the fiber $f^{-1}(y)$ is 
nonempty, and each of its irreducible components has 
dimension $\dim X-\dim Y$.
\end{theorem}

\begin{example}
In Theorem~\ref{thm:dominance-function-fields}, if $\dim X=\dim Y$,
then $K(X)/K(Y)$ is finite
algebraic, since it is a finitely generated field extension of
transcendence degree zero. Thus a dominant morphism from an irreducible
curve $D$ to $\A^1_K$ gives a finite extension $K(D)/K(t)$, with the function field of $\A^1_K$ as $K(t)$.
\end{example}
Now we recall the notions of integrality and integral closures, and state some geometric consequences.

\begin{definition}
\label{def:integral-dependence}
For a ring extension of commutative rings with 1 
$A\subseteq B$, an element $b\in B$ is
\emph{integral over $A$} if
\[
 b^d+a_{d-1}b^{d-1}+\cdots+a_0=0
 \quad\text{for some }d\geq1\text{ and }a_i\in A.
\]
The extension is \emph{integral} if every element of $B$ is integral
over $A$. The elements integral over $A$ form its \emph{integral closure}
in $B$. An integral domain $A$ is \emph{integrally closed} if its integral closure
in $\Frac(A)$ equals $A$; see~\cite[Chapter~5, pp.~59--63]{AM69}.
\end{definition}

A morphism $f:X\to Y$ of affine varieties is 
\emph{finite} when its coordinate
ring $\C[X]$ is a finite $\C[Y]$-module via $f^*$.
\begin{remark}[{\cite[Proposition 5.1 and Corollary 5.2]{AM69}}]\label{rem:finite-integral}
Every finite ring extension $A\subseteq B$ is integral.
Furthermore, if $B$ is generated by adjoining finitely many integral elements over $A$, then $B$ is finite as an $A$-module. 
\end{remark}
We recall the following
standard geometric properties of finite morphisms.

\begin{theorem}[{\cite[Tags 01WM,02LS,0ECG,0F32,0BQ3]{Stacks26}}]
\label{thm:finite-morphisms-geometric}
Let $f:X\to Y$ be a finite morphism of affine varieties
over $K$. Then the map $f$ is closed (maps closed subspaces to closed subspaces) and every fiber is finite. 
So $\dim(f(Z))=\dim(Z)$ for each closed $Z\subset X$. 
In particular,
if $f$ is dominant, then 
it is surjective.

Furthermore, if $X$ and $Y$ are irreducible, with 
$K[Y]$ integrally closed in $K(Y)$, and $f$ is
dominant, then $f$ is open, i.e., $f$ maps open 
subsets to open subsets.
\end{theorem}

The following useful consequence is an application of 
the above theorem to $Y=\A^d_K$.

\begin{lemma}\label[lemma]{lem:horizontal}
Let $X$ be an irreducible affine variety of dimension $d\geq1$ and let
$\pi:X\to\A^d_K$ be finite and dominant. 
For any affine line $\A_K^1\simeq L\subseteq\A^d_K$, let 
$C\coloneqq \pi^{-1}(L)$, and let $W_1,\ldots, W_k$ be 
the irreducible components of $C$.  Then $\pi$ maps each
$W_i$ surjectively onto $L$.
In particular, for $b\in L$,
every point of $\pi^{-1}(b)$ is in one of $W_i$.
\end{lemma}
\begin{proof}
The closed space $C$ is a curve because $\pi$ is 
finite and dominant, hence surjective, and 
$\dim(C)=\dim(\pi(C))=\dim(L)=1$. 

Since $K[\A^d]\simeq K[x_1,\ldots, x_d]$ is just the polynomial ring in $d$-variables, and so is integrally closed in $K(x_1,\ldots, x_d)$, the map $\pi$ is also open, and
so is its restriction to $C$.
We show without loss of generality that $\pi$ restricted to $W_1$ is surjective. 
Let $U$ be the open 
set obtained after removing the irreducible closed sets
$W_2,\ldots, W_k$ from $X$. 
Then $\pi(W_1) \supseteq \pi(U\cap C)=\pi(U)\cap L\ne \emptyset$ is
open in $L$. Also, $\pi(W_1)$ is closed, as $W_1$ is a closed set. Since $L$ is irreducible and $\pi(U)\cap L$ is a non-empty open set, its closure is $L$. Hence $\pi(W_1)=L$. 
\end{proof}

\subsection{Image closures and approximate vanishing}
\label{sec:prelim-approximation}
\label{sec:prelim-aps-geometry}
From now on we assume $K=\C$, however one may as well 
take $K$ to be any algebraically closed field.
Recall from Theorem \ref{thm:dominance-function-fields}
and Example \ref{ex:example-irred}
that for a polynomial map 
$F=(f_1,\ldots,f_m):\A_\C^n\to\A_\C^m$ and $Y=\overline{F(\A^n_\C)}$, $Y$ is irreducible of dimension $\dim(Y)=
\trdeg(f_1,\ldots, f_m)$ with $I_Y$ being the 
annihilator ideal $\Ann(F)$ of $f_1,\ldots, f_m$.

Let
\[
 \K=\bigcup_{e\geq1}\C((t^{1/e})),\qquad
 \mathcal O_\K=\{a\in K:\val_t(a)\geq0\},\qquad
 \mathfrak m_\K=\{a\in K:\val_t(a)>0\}.
\]
Here $\K$ is the field of formal Puiseux series that is
algebraically closed, and $\val_t=\ord_t:K\to \Q$, is the valuation map, defined as follows: for an element $\alpha=t^{e}(c_0+c_1t^{e_1}+\ldots)$
with $c_0\ne 0$, we define $\val_t(\alpha)=e\in\Q$, and 
set $\val_t(0)=+\infty$.

\begin{definition}[Approximate vanishing and APS]
  For a tuple $F=(f_1,\ldots,f_m)$ of polynomials in $\C[x_1,\ldots, x_n]$, let
\[
 \AV(F)=\{\beta\in \K^n:\val_t(f_i(\beta))>0
                      \text{ for every }i\}.
\]
We say that the system (equivalently the map) $F$ has
APS if $\AV(F)\ne \emptyset$.
\end{definition}
The following is the statement of weak approximative 
nullstellensatz (in the sense of \cite{ADPS26}).
\begin{theorem}[{\cite[Lemma~4.1 and Theorem~4.2]{GSS19}}]\label{thm:WAN}
Let $F$ be as above. Then 
\begin{equation}\label{eq:wan}
 \AV(F)\ne\varnothing
 \quad\Longleftrightarrow\quad
 0\in\overline{F(\A^n_\C)}
 \quad\Longleftrightarrow\quad
 P(0)=0\text{ for all }P\in\Ann(F).
\end{equation}
\end{theorem}
The definition of $\AV(F)$ in \cite{GSS19} is given 
in terms of $\beta\in \C((t))^n$, however it is enough
to assume  $\beta\in (\K)^n$.
Indeed, a Puiseux witness $\beta(t)\in\K^n$ of $\AV$ 
will belong to $\C((t^{1/e}))^n$, for some large enough
$e$, and so $\beta(t^e)
\in \C((t))^n$ gives a witness in the sense of
\cite{GSS19}.

In the following section we recall the standard results that reduce the general setting to the case when $\trdeg(f_1,\ldots, f_m)=n$ and $m=n+1$.

\subsection{Random hyperplane sections and random projections}
\label{sec:rewrite-reduction}

Set $r=\dim Y$ and $d=\max(1,\deg f_1,\ldots,\deg f_m)$.
We assume that $f_i$ are non-constant, otherwise there
is nothing to prove. If $r=m$, then $Y=\A^m$ and APS 
trivially holds. Hence we assume from now on that
$1\leq r<m$.

We state the random hyperplane-section and projection results in the forms used in this paper.
In each statement, coefficients are sampled independently and uniformly
from a finite set $B\subset\Q$. We recall the following
standard result mentioned in \cite{AGS26}. It estimates
the probability with which a random hyperplane 
intersection drops the dimension of a given affine 
algebraic set $W$. This is in some ways, 
explicitization of the classical result of Bertini 
on hyperplane intersections. Successively intersecting 
with $\dim(W)$ many random hyperplanes, one can reduce
the dimension of $W$ to $0$. 

\begin{lemma}[{\cite[Lemma~4.9]{AGS26}}]
\label{lem:rewrite-random-sections}
Let $W\subseteq\A^a_\C$ be an affine algebraic set of 
degree $\Delta$.
A random affine linear equation reduces its dimension to 
$\dim(W)-1$ with
probability at least $1-2\Delta/|B|$. Consequently, 
$\dim(W)$ successive random affine
linear equations have a finite nonempty intersection with $W$ with
probability at least $1-2\Delta\cdot\dim(W)/|B|$.
\end{lemma}
Degree does not increase under hyperplane sections.
The following result reduces the number of outputs to $r+1$ by a
random linear projection, preserving both the image dimension and APS.

\begin{theorem}[Geometric form of {\cite[Theorem~4.6]{GSS19}}]
\label{thm:rewrite-random-projection}
Let $F=(f_1,\ldots,f_m)$ be polynomials of degree at most $d$ and
transcendence degree $r<m-1$. Choose $c_{ij}\in B$ and put
$\widetilde f_i=\sum_{j=1}^m c_{ij}f_j$, $1\leq i\leq r+1$.
With probability at least $1-(r+1)d^r/|B|$,
\[
 \trdeg_\C(\widetilde f_1,\ldots,\widetilde f_{r+1})=r,
 \qquad
 0\in\overline{F(\A^n_\C)}
 \ \Longleftrightarrow\
 0\in\overline{\widetilde F(\A^n_\C)}.
\]
\end{theorem}
Combining these two results reduces deciding APS to $m=n+1$ and
transcendence degree $n$ case, which we record as the following proposition.

\begin{proposition}
\label{prop:rewrite-normal-form}
With the notation above, impose $n-r$ random affine linear equations
on the variables and parametrize their intersection. If $m>r+1$,
then take $r+1$ random linear combinations of the outputs.
With probability at least
\[
 1-\frac{2(n-r)d^m+(r+1)d^r}{|B|},
\]
this construction produces a map
$\widetilde F:\A^r_\C\to\A^{r+1}_\C$ of transcendence degree $r$ such that
$F$ has APS if and only if $\widetilde F$ has APS.
\end{proposition}
\begin{proof}
Let $Y=\overline{F(\A^n_\C)}$, and
$E=\overline{\C(Y)}$, the algebraic closure of function field of $Y$.
Recall that the map $F$ is a map of irreducible algebraic sets and is dominant on its image 
$F(\A^n_\C)$, Theorem \ref{thm:dominance-function-fields} 
gives that a general fiber has dimension $n-r$,
and degree at most $d^m$ by B\'ezout's
inequality.  This corresponds to saying that in $\A^n_E$ the system of equations $f_i-y_i$, for all $1\leq i\leq m$ define an affine algebraic set of dimension $n-r$ over $E$. 
Since $\C^n\subset \A^n_E$ is Zariski dense we can as well sample the coefficients of our random hyperplanes from $\Q$.
Applying Lemma~\ref{lem:rewrite-random-sections} to this affine algebraic set over $E$ gives that 
$n-r$ random hyperplanes cut it to a finite nonempty set with failure probability at most 
$2(n-r)d^m/|B|$.
So the map $F$ restricted to the $r$-plane given by the 
intersection of $n-r$ random hyperplanes is still 
dominant onto $F(\A^n_\C)$. So we obtain a map 
$\A^r_\C\to \A^m_\C$, that preserves APS. 

Theorem~\ref{thm:rewrite-random-projection}, then 
gives $r+1$ outputs while preserving APS and 
transcendence degree. The sum of the failure probabilities, then
give the claimed bound.
\end{proof}

Both operations preserve polynomial circuit size. If the input circuit
length is $L$, then $\log d$ and the logarithms of the displayed degree
bounds are polynomial in $L$. Hence an integer grid $B$ of polynomial
coefficient bit length makes the failure probability a fixed small
constant. The image dimension can be found as in
Lemma~\ref{lem:dimension}. Thus the geometric construction may be
carried out with the reduced map, renamed $F$, and dimensions $n,m$
satisfying
\begin{equation}\label{eq:rewrite-normal-form}
 m=n+1,\qquad \dim Y=\trdeg_\C(f_1,\ldots,f_m)=n\geq1.
\end{equation}

In this reduced case, we use the following standard form of linear
Noether normalization.

\begin{lemma}[{\cite[Proposition~2.42 and Lemma~2.43]{Mil24}}]
\label{lem:rewrite-linear-nn}
Let $Y\subseteq\A^{n+1}_\C$ be an irreducible affine variety of dimension
$n\geq1$. For a linear map
$\pi:\A^{n+1}_\C\to\A^n_\C$ given by $\pi(y_1,\ldots,y_{n+1})=(y_1-c_1y_{n+1},\ldots, y_n-c_ny_{n+1})$,
for general $(c_1,\ldots, c_n)\in \A^n_\C$ 
the restriction $\pi:Y\to\A^n_\C$ is finite
and surjective.
\end{lemma}

By Remark~\ref{rem:finite-integral}, the restriction
$\pi:Y\to\A^n_\C$ induces an injective ring homomorphism
\[
\pi^*:\C[u_1,\ldots,u_n]
\longrightarrow
\C[Y]=\C[y_1,\ldots,y_{n+1}]/I_Y,
\qquad
u_i\longmapsto (y_i-c_i y_{n+1})\bmod I_Y.
\]
Moreover, $\C[Y]$ is a finite module over the image of
$\pi^*$. Consequently, each coordinate function
$y_j|_Y$, for $1\leq j\leq n+1$, is integral over this image.

The following gives the failure probability of a random choice of $c\in\A^n_\C$ in Lemma \ref{lem:rewrite-linear-nn}.

\begin{remark}\label{rem:nnl-failure}
Since $Y$ is an irreducible affine hypersurface, it is given by 
the vanishing of a single irreducible polynomial $P$. Let $\delta=\deg P$, and
$P_\delta$ be the highest-degree homogeneous part of $P$.
By the proof of \cite[Lemma~2.43]{Mil24} the 
linear map $\pi: \A^{n+1}_\C\to \A^{n}_\C$ fails to give 
a Noether normalization if the chosen 
$(c_1,\ldots, c_n)\in \A^n_\C$ satisifes 
$P_\delta(c_1,\ldots,c_n,1)=0$. 
If $c_i$ are drawn independently and uniformly 
from a set $B\subset \Q$, then Schwartz-Zippel lemma
gives the failure probability of $\pi$ being a Noether 
normalization being at most $\delta/|B|$.
\end{remark}

\subsection{Complexity theoretic inputs}

\paragraph{Input size.}
For both problems, $L$ denotes the total size of the 
 input circuits, including binary rational
constants. Gates perform addition, subtraction, and 
binary multiplication. Formal degrees are computed by
taking the maximum at addition gates and the sum at 
multiplication gates; they have polynomial-length 
binary representations. Output degrees and coefficient
bit lengths are at most $2^{\poly(L)}$. We never expand
these outputs. After clearing rational denominators gate by gate,
Kronecker substitution reduces coefficient-bit access to
$\mathsf{BitSLP}\in\CH$~\cite[Theorem~4.1]{ABKPM09}.
The implicit coefficient table has size $2^{\poly(L)}$; hence
\cite[Corollaries~5.13 and~6.4, Lemma~6.5]{AGS26} give
resultant-coefficient access in $\CH$ for polynomially many such
forms, without a degree-two assumption.
A sparse polynomial with binary exponents
converts to a circuit in polynomial time by repeated 
squaring. An input rational target is handled by 
subtracting its coordinates at the output gates. Dummy
variables and a zero generator allow us to assume 
$n,m\geq1$ without changing either problem.

\paragraph{The Counting Hierarchy.}
Set $C_0\mathsf P=\mathsf P$ and
$C_{i+1}\mathsf P=\mathsf{PP}^{C_i\mathsf P}$. Then
$\CH=\bigcup_{i\geq0}C_i\mathsf P$ 
\cite{koiran2011interpolation}. A fixed number of
polynomially bounded existential or universal 
quantifiers over a language at a fixed level stays 
in $\CH$. The hierarchy is closed under complement 
and polynomial-time Turing reductions to a language at
a fixed level. A randomized
polynomial-time computation with such an oracle also
stays in $\CH$, since $\mathsf{BPP}^A\subseteq
\mathsf{PP}^A$. A fixed finite collection of these 
oracles can be combined at one fixed level.

\paragraph{Trailing coefficients.}
For a nonzero polynomial $R(t,z)=\sum_{i\geq0}t^iR_i(z)$ over a domain, let
\[
 \TT_t R=R_e(z),\qquad e=\min\{i:R_i\ne0\}.
\]
Thus $\TT$ extracts a coefficient, not the coefficient times $t^e$. Trailing coefficients are multiplicative, and nested operations are read from right to left. The coefficients are polynomials in all the remaining variables. Consequently, evaluation at $s=0$ generally does not commute with a trailing-coefficient operation.

\section{The geometry of APS}\label{sec:geometry-rewrite}
In this section, we discuss the geometric construction behind reducing the problem to multivariate resultants. 
We follow the discussion in Section \ref{sec:aps_overview}.
Let $0_m$, $0_n$, and $0$ be the origins of $\A^m_\C$, $\A^n_\C$, and $\A^1_\C$, respectively, and
let $Y=\overline{F(\A^n_\C)}$. Then
Proposition~\ref{prop:rewrite-normal-form} reduces the problem to
\[
 F=(f_1,\ldots,f_m):\A^n_\C\longrightarrow\A^m_\C,\qquad
 m=n+1,\qquad \dim Y=n,\qquad .
\]
Let $x_1,\ldots, x_n$ be the input coordinates (variables) of $F$ and $y_1,\ldots, y_m$ be the coordinates on $\A^m_\C$, i.e., the output variables of $F$. For notational simplicity we will refer to
$x_1,\ldots, x_m$ and $y_1,\ldots, y_m$ by $x$ and $y$ coordinates, respectively.
Let $d\geq\max(2,\deg f_1,\ldots,\deg f_m)$ be the formal-degree
bound from the input-size discussion; thus $d\leq2^{\poly(L)}$.
Recall the field of Puiseux series $\K$ of 
Section~\ref{sec:prelim-aps-geometry}.
The aim is to construct a system of $m+n$ homogeneous 
polynomials in $m+n+1$ variables over $\K$ such that 
affine solutions of this system over $\K$ exactly 
capture if
$0_m\in Y$.

\subsection{The graph and a finite family along a line}
\label{sec:rewrite-line}

Let $\Gamma_F\subset\A^n_\C\times \A^m_\C$ be the graph
of the polynomial map $F$, and let $p:\Gamma_F\to Y$ 
given by $(x,y)\mapsto y$ be the projection map to the 
second component. Then 
\[
 \Gamma_F=V(y_1-f_1(x),\ldots,y_m-f_m(x)).\]
 Recall that $\Gamma_F\cong\A^n_\C$, and $p$ is a dominant 
 morphism of irreducible affine algebraic spaces by 
 the definition of $Y$.
Theorem~\ref{thm:dominance-function-fields} gives a nonempty open
$U\subseteq Y$ over which the fibers of $p$ are finite 
and nonempty. 
Choose a finite surjective linear 
projection $\pi:Y\to\A^n_\C$ using Lemma~\ref{lem:rewrite-linear-nn} ($\pi$ gives a Noether normalization of $Y$). Let $u_1,\ldots, u_n$
be the coordinates on the image of $\pi$, denoted by $u$.

It will be important to choose $U$ as a 
complement of a closed set of singly exponential degree.
The composition $ \pi\circ F:\A^n_\C\to \A^n_\C$ is dominant, 
so Theorem \ref{thm:dominance-function-fields}
shows that $x_i$ is algebraic over the field
generated by the polynomials $\pi\circ F$, equivalently over $\C(u_1,\ldots, u_n)$.
Indeed, Perron's bound~\cite[Section~1, p.~2]{GSS19} 
gives relations $Q_i(\pi(F),x_i)=0$ with 
$\deg Q_i\leq d^n$. Let $b$ be the product of the 
leading coefficients of $Q_i(u,T)$ in $T$, which are 
nonzero.
Outside the algebraic set $b=0$, the map $p$ is finite 
and surjective by Section~\ref{sec:prelim-integrality}.
So we can choose $U$ as the complement of the closed 
set defined by $b=0$ of degree $nd^n$. 
B\'ezout along with $\deg(Y)=\deg(\Ann(F))\leq d^n$ gives
$\deg(Y\setminus U)\leq n d^{2n}$. We summarize this 
in the following proposition.
\begin{proposition}\label{prop:U-singly-exp}
Let $U\subset Y$ be as above. Then $U$ can be chosen
as a complement of a closed set of degree at most $nd^{2n}$. 
\end{proposition}
By Theorem~\ref{thm:finite-morphisms-geometric}, $\pi$ 
is closed, so $\pi(Y\setminus U)$ is closed of dimension
$\dim(Y\setminus U)<n$. Hence let $V$ be the open set 
\[
 V=\A^n_\C\setminus\pi(Y\setminus U).
\]
Thus every point of $Y$ above $V$ belongs to $U$.

Choose a nonzero $\alpha\in\A^n_\C$ such that
$L=\{t\alpha:t\in\C\}$ meets $V$, and identify $L$ with $\A^1_\C$, whose coordinate ring is given by $\C[t]$. Let
\[
 C\coloneqq Y\cap\pi^{-1}(L),\qquad
 D\coloneqq \overline{(\pi\circ p)^{-1}(L\cap V)}\subseteq\Gamma_F.
\]
Let $\beta\coloneqq \pi\circ p$. Then $\beta$ is finite
and surjective over $V$.
Since $V$ is an open subset of $\A^n_\C$, so it is 
integrally closed in its field of fractions which is 
the same as $\C(u_1,\ldots, u_n)$. Hence by 
Theorem~\ref{thm:finite-morphisms-geometric} $\beta$ is 
open. The following proposition shows that the map
$\beta$ is dominant on each component of $D$. 
\begin{proposition}\label{prop:D-is-dominant}
Let $D$, $\beta$, $L$ and $V$ be as above, and let $D'$
be an irreducible component of $D$. 
Then $\beta|_{D'}:D'\to L$ is dominant, and $K(D')/K(L)$
is a finite extension. 
\end{proposition}
\begin{proof}
  The proof is similar to Lemma~\ref{lem:horizontal}. 
  Recall that $\beta$ is open. We apply the component
  argument as in 
  the proof of Lemma~\ref{lem:horizontal}, over
$L\cap V$.  Each component of 
$\beta^{-1}(L\cap V)$ maps onto $L\cap V$ and has dimension one. On taking closures, one obtains that 
each irreducible component $D'$ of $D$ is a curve that
dominates $L$, and so 
the function field of $D'$, $\C(D')$ is a finite
extension of fields over $\C(t)\simeq \C(L)$ by Theorem~\ref{thm:dominance-function-fields}.
\end{proof}
Over $\C(t)$ the required equations of $D$ are simply
\begin{equation}\label{eq:rewrite-square}
 \begin{cases}
 g_i=y_i-f_i(x)=0 & (1\leq i\leq m),\\
 g_{m+j}=\pi_j(y)-\alpha_jt=0 & (1\leq j\leq n).
 \end{cases}
\end{equation}
There are $n+m$ equations in $n+m$ variables. Their solution set is
finite and nonempty for $\alpha t\in V$.
When viewed over $\C(t)$, the Proposition \ref{prop:D-is-dominant} says that 
each coordinate $x_i$ or $y_i$ is algebraic 
over $\C(t)$, and so a solution to the system \ref{eq:rewrite-square} over $\C(t)$ has coordinates that  are algebraic over $\C(t)$.
In the next section, we see that the $y$-coordinates are in fact integral, and how one can view them as elements of $\K$. The following example illustrates the above discussion.
\begin{example}\label{ex:line-slice}
For $F(x_1,x_2)=(x_1,x_1x_2-1,x_1^2x_2-1)$, the image closure is
$Y=V(y_3-y_1(y_2+1)+1)$. Take the $y_1,y_2$-plane $H=\{y_3=0\}$,
identified with $\A^2$, and project by $\pi(y)=(y_1,y_2)$. This is an
isomorphism on $Y$, hence a finite Noether normalization. The line
$L\subset H$ given by $y_1=y_2=t$ imposes $x_1=t$ and $x_1x_2=t+1$.
The projection and line restriction therefore give a finite generic fiber
whose coordinates lie in $\C(t)$. We can choose our open set $V$ to be $\A^2_\C\setminus \{y_1=0\}$.
The generators $g_i$ of $I_D$ are given by 
\begin{equation}\label{eq:overview-example-square}
 g_1=y_1-x_1,\ g_2=y_2-x_1x_2+1,\ g_3= y_3-x_1^2x_2+1,\ 
g_4= y_1-t,\ g_5=y_2-t.
\end{equation}
Above $V$, the solution set over $\C(t)$ of the system \eqref{eq:overview-example-square} is finite. 
\end{example}

\subsection{Specializing the finite generic fiber}
\label{sec:rewrite-lifting}

Put $Y_0=Y\cap\pi^{-1}(0_n)$. This set is finite and nonempty by
Theorem~\ref{thm:finite-morphisms-geometric}.
Lemma~\ref{lem:horizontal} already ensures that every $y_0\in Y_0$
lies on a component of $C$ dominating $L$. 
First we recall the following standard fact about 
algebraic embeddings and integral elements. 

\begin{lemma}[{\cite[Tags~09GU,~00IC]{Stacks26}}]
\label{lem:rewrite-branches}
The natural inclusion $\C(t)\hookrightarrow\C((t))$ extends to an embedding
$M\hookrightarrow\K=\overline{\C((t))}$, for each Galois extension $M$ of $\C(t)$. 
Every such embedding maps the integral closure of $\C[t]$ in $M$ into $\mathcal O_\K$, 
and $\mathcal O_\K$ is integrally closed. 
\end{lemma}

Recall from Proposition \ref{prop:D-is-dominant} that for each irreducible component $D'$ of $D$,
$\C(D')$ is a finite extension of $\C(t)$. In fact, 
the coordinates $(x(t),y(t))$ of a point on $\C(D')$ 
are algebraic over $\C(t)$. By Lemma 
\ref{lem:rewrite-branches} we consider them as elements 
of $\K$. The following proposition gives the 
valuation theoretic properties of the 
$y$-coordinates.
\begin{proposition}
\label{prop:rewrite-lift}
Every solution $(x(t),y(t))\in\K^{n+m}$ of
\eqref{eq:rewrite-square} satisfies $\val_t(y_i(t))\geq0$,
i.e., $y_i(t)\in \mathcal O_\K$, and
$y(0)\in Y_0$. Conversely, every point of $Y_0$ occurs as such a
residue.
\end{proposition}
\begin{proof}
Since $\pi$ is finite, Remark~\ref{rem:finite-integral} shows that
each $y_i$ is integral over $\C[u_1,\ldots,u_n]$ via $\pi^*$.
Substituting $u_j=\pi_j(y)=\alpha_jt$ gives monic equations over $\C[t]$,
so $\val_t(y_i(t))\geq0$. Taking constant terms in the equations of
$Y$ and in $\pi(y)=t\alpha$ gives $y(0)\in Y_0$.

Conversely, fix $y_0\in Y_0$ and choose an irreducible
component $W$ of $C$ through $y_0$ dominating $L$,
by Lemma~\ref{lem:horizontal}.
Since $p(D)$ contains the part of $C$ above $L\cap V$,
some irreducible component $D'$ of $D$ dominates $W$.
Let $E=\C(D')$ and take a finite Galois extension
$M/\C(t)$ containing $E$. Since $W\to L$ is finite,
the coordinate ring $\C[W]$ is contained in the
integral closure $B$ of $\C[t]$ in $M$.

Recall from the classical geometric points and maximal ideals correspondence that there is a maximal ideal $\mathfrak{m}_{y_0}$ of $\C[W]$ whose vanishing set is exactly the point $y_0$. The ideal $\frak m_{y_0}$ is generated by image of $y_i-y_{0,i}$ in $\C[W]$ by Hilbert Nullstellensatz. 
By the lying-over
theorem~\cite[Theorem~5.10]{AM69}, there exists a
maximal ideal $\mathfrak n$ of $B$ lying above
$\mathfrak m_{y_0}$.

By Lemma~\ref{lem:rewrite-branches}, choose an embedding
$\rho:M\hookrightarrow\K$ over $\C(t)$.
Since $B$ is integral over $\C[t]$, we have
$\rho(B)\subseteq\mathcal O_\K$.
The embedding $\rho$ therefore determines a maximal
ideal of $B$ above $(t)$ because $\frak m_{y_0}$ maps to an ideal above $(t)$ in $\rho(\C[W])$.
Since the maximal ideals in a finite
Galois extension are Galois conjugates ~\cite[Tag~09EA]{Stacks26},
we may compose $\rho$ with a Galois automorphism
so that the resulting embedding satisfies
\[
\val_t\bigl(\rho(y_i)-y_{0,i}\bigr)>0
\qquad (1\leq i\leq m).
\]
Applying this embedding to the coordinate functions
of $D'$ gives a solution $(x(t),y(t))\in\K^{n+m}$
of~\eqref{eq:rewrite-square} satisfying
$y(0)=y_0$. 
\end{proof}
In Example \ref{ex:line-slice}, the curve $C$ is itself irreducible, and for $t\in V=\A^1_\C\setminus \{0\}$, the coordinates of 
solutions over $\C(t)$ of the
system \ref{eq:overview-example-square} are
\[(x_1(t),x_2(t),y_1(t),y_2(t),y_3(t))=
 \left(t,\frac{t+1}{t},t,t,t^2+t-1\right)\in D,
\]
of whom the $y$-coordinates $t$, $t$, $t^2+t-1$ are integral over $\C[t]$. Furthermore, $Y_0=\{(0,0,-1)\}$, and one can directly see that $y_i-y_{0,i}$ is $t, t, t^2+t$, for $i=1,\ 2,\ 3$, respectively. So for all $i$ $\val_t(y_i-y_{0,i})=1$.

\subsection{Relative closure and homogenization of all coordinates}
\label{sec:rewrite-homogenization}

Take the closure $\overline D\subseteq\PPj^n_\C\times C$
via the natural embedding
$\A^n_\C\times C\to\PPj^n_\C\times C$.
The closure $\overline D$ is the union of the closures
$\overline{D'}$ of the irreducible components $D'$ of $D$.
By Proposition~4.2, each $D'$ is an irreducible curve.
Since $D'$ is locally closed, it is open and dense in
$\overline{D'}$. Thus $\overline{D'}\setminus D'$ is a
proper closed subset of an irreducible curve and hence
consists of finitely many points
(see \cite[Chapter~I, \S 8, p.~56]{Har77}).
Moreover, the irreducible components of $\overline D$
are precisely the closures of the irreducible components
of $D$.

Since $\PPj^n_\C$ is complete, the projection
$\PPj^n_\C\times C\to C$ is closed
\cite[Definition~7.1 and Theorem~7.22]{Mil24}. Extend the map $p$ to $\overline{D}$, and call it again $p$. Thus $p(\overline D)$
is closed and contains the dense part of $C$ above $L\cap V$, so
$p(\overline D)=C$.
Since $C\to L$ is finite and surjective, the composition 
$\overline\beta=\pi\circ p:\overline{D}\to L$ is surjective with finite fibers.

Moreover, recall by Proposition~\ref{prop:D-is-dominant}, every irreducible
component of $D'$ of $D$ is a curve dominating $L$, and we obtain $\overline D$ by closing each irreducible component of $D$. Hence for each irreducible component $ D'$ of $D$,
the map $\overline\beta|_{\overline{D'}}:\overline {D'}\to L$ is surjective with finite fiber.

For $t\alpha\in L$, define the finite set
\[\overline{S_t}\coloneqq \overline{\beta}^{-1}(t\alpha).\]
Then 
\begin{equation}\label{eq:rewrite-special-fiber}
  p(\overline{S_0})=Y_0,\qquad
0_m\in Y\ \iff\ 0_m\in p(\overline{S_0}).
\end{equation}
By Proposition~\ref{prop:rewrite-lift}, the points of $D$ over
$\K$ are represented by solutions $(x(t),y(t))$ of
\eqref{eq:rewrite-square}, with $y(t)\in\mathcal O_\K^m$.
Clearing denominators in the projective $x$-coordinates gives
\[
[1:x_1(t):\cdots:x_n(t)]
=
[t^a:t^ax_1(t):\cdots:t^ax_n(t)],
\]
where $a\in\Q_{\geq 0}$ is chosen so that all coordinates have
nonnegative valuation and at least one has valuation zero.
Taking constant terms gives a point of $\overline D$ above
$t=0$.

Conversely, let $q\in\overline{S_0}$, and choose an
irreducible component $\overline{D'}$ of $\overline D$ containing
$q$. By Proposition~\ref{prop:D-is-dominant}, $\overline{D'}$
dominates $L$.

By the local Newton--Puiseux parametrization theorem
for algebraic curves
\cite[Tags~0C45,~0C0S]{Stacks26} (see also \cite[Lecture~14, p.~104]{Abhyankar90}),
there is a Puiseux branch $\tilde{q}(s)$ of $\overline{D'}$ through $q$,
parametrized by $s=t^{1/e}$ for some $e\geq1$,
such that the projective
$x$-coordinates of $\tilde{q}(s)$ are
\[
[x_0(s):x_1(s):\cdots:x_n(s)],
\qquad x_i(s)\in\C[[s]],
\]
with $\tilde{q}(0)=q$.
Since the component intersects $D$ in an open set of $D$,
we have $x_0(s)\not\equiv0$.
Consequently, $x_i(s)/x_0(s)\in\K$,
and together with the corresponding $y$-coordinates
these give a solution of~\eqref{eq:rewrite-square}.
Clearing denominators in its projective
$x$-coordinates recovers $\tilde{q}$.
We summarize this discussion as the following proposition.
\begin{proposition}\label{prop:Dbar-has affine-coordinates}
Every point $q\in\overline{S_0}$ is the specialization of a Puiseux family
$\tilde q(t)\in\overline{S_t}$ with homogeneous coordinates
$[x_0(t):\cdots:x_n(t)]$, where for all $i$, $x_i(t)\in\mathcal O_\K$,
and $x_0(t)\not\equiv 0$.
Therefore, $\overline{S_t}\subseteq\{x_0\neq 0\}\simeq\mathbb A^n_\K$,
so $\overline{S_t}$, viewed as a projective algebraic set over $\K$,
is affine.
\end{proposition}

Returning back to the Example \ref{ex:line-slice},
one sees that the curve 
$D$ is given by $(t,(t+1)/t,t,t,t^2+t-1)$. The $x$ coordinates are defined over $\C(t)$. Viewing it in projective coordinates, we 
obtain that $\overline{D}$ 
is given by \[\left(\left(1:t:\frac{t+1}{t}\right),(t,t,t^2+t-1)\right)=\left((t:t^2:t+1),(t,t,t^2+t-1)\right)\subset \PPj^2_\C\times C.\]
Furthermore, $x_0\ne 0\in \mathcal O_\K.$

This relative closure explains the geometry, but does not provide
a homogeneous system with the same number of homogeneous equations as the number of homogenized variables, which is needed for a resultant. Homogenizing
only in $x$ would still require the equations of $C$, which we do
not know. 

In Example~\ref{ex:line-slice}, the homogenized graph
equations in just the $x$-coordinates allow $[x_0:x_1:x_2]=[0:0:1]$ with arbitrary $y_3$,
whereas $C$ requires $y_3=t^2+t-1$.
In particular, at $t=0$ just
homogenizing in $x$-coordinates will not be able to 
distinguish if $0\in Y$ unless the equations defining $C$
are known, which is something we want to avoid.

We observe, that if we could homogenize the system \ref{eq:rewrite-square} with respect to all $x$ and $y$ variables, then we will get a homogeneous system amenable to the resultant method of \cite{AGS26}.
So we use the map 
\[
 \iota:\PPj^n_\K\times\A^m_\K\longrightarrow\PPj^{n+m}_\K,\qquad
 ([x_0:\cdots:x_n],y)\longmapsto
 [x_0:\cdots:x_n:x_0y_1:\cdots:x_0y_m],
\]
and show that we do not lose the information we were looking for. 
Observe that on $x_0\ne0$, division by $x_0$ recovers both $x$ and $y$.
Thus the Puiseux branches in Proposition \ref{prop:Dbar-has affine-coordinates} retain their output coordinates over $\K$;
information is lost only when $x_0$ is identically zero.

Once again for the case of Example~\ref{ex:line-slice}, the relative to $L$ branch and
its image under $\iota$ are
\[
 \left([t:t^2:t+1],(t,t,t^2+t-1))\right)
 \longmapsto
 v(t)=[t:t^2:t+1:t^2:t^2:t(t^2+t-1)].
\]
Its first coordinate is nonzero over $\K$, so we can 
extract $y$-coordinates by factoring out $x_0(t)=t$, which is not identically 0, from 
the last three coordinates $t^2,t^2,t(t^2+t-1)$.
If we specialize to $t=0$, then the last three $y$-coordinates become $(0,0,0)$, and so the projection to the $y$-coordinates will contain $(0,0,0)$, but as noted earlier it is not a point in  $Y_0$.
This also shows why it is important to first work over $\C[t]$, and then specialize later to $t=0$, instead of specializing first to $t=0$.

We set $N=n+m$, $z=(x,y)$, and homogenize using the formal degree bounds over $\C[t]$. Let
\begin{equation}\label{eq:rewrite-homogeneous}
 G_i(z_0,z,t)=z_0^{d_i}g_i(z/z_0,t),\qquad
 \max(1,\deg_z g_i)\leq d_i\leq d,\quad 1\leq i\leq N.
\end{equation}
These are $N$ homogeneous equations in $N+1$-variables over
$\C[t]$. On the chart $z_0\ne0$ over $\K$, they give exactly the
finite solution set of~\eqref{eq:rewrite-square}.

\begin{theorem}
\label{thm:rewrite-square-aps}
The map $F$ has APS if and only if there are
$z_0(t),\ldots,z_N(t)\in\K$ satisfying
\begin{equation}\label{eq:rewrite-valuative-test}
 \begin{gathered}
 z_0(t)\ne0,\qquad \min_{0\leq j\leq N}\val_t(z_j(t))=0,\\
 G_i(z_0(t),\ldots,z_N(t))=0\quad(1\leq i\leq N),\\
 \val_t\!\left(\frac{z_{n+i}(t)}{z_0(t)}\right)>0
                         \quad(1\leq i\leq m).
 \end{gathered}
\end{equation}
\end{theorem}
\begin{proof}
If $0\in Y$, Proposition~\ref{prop:rewrite-lift} supplies an affine
solution with $y(0)=0$, so $\val_t(y_i(t))>0$. Clearing the denominators
in homogeneous coordinates $(1:x_1(t):\ldots:y_m(t))$ 
with a power of $t$ such that the minimum of $\val_t(x_i)$ 
or $\val_t(y_i)$ is zero, one obtains $z_0$ with required
properties.
Conversely, dehomogenizing a tuple in~\eqref{eq:rewrite-valuative-test}
gives $y_i(t)=f_i(x(t))=z_{n+i}(t)/z_0(t)$, so 
$\val_t(y_i(t))>0$ and  
$y(0)=0$ and $0\in Y$.
\end{proof}

\section{An approximative characteristic polynomial}
\label{sec:characteristic-rewrite}

We now construct the polynomial $R_\ell(T)$ from the 
proof idea in
Section~\ref{sec:aps_overview}. Keep the reduced map and the square
system~\eqref{eq:rewrite-square}, i.e., the number of variables are the same as the number of equations, and homogeneous
system~\eqref{eq:rewrite-homogeneous} from
Section~\ref{sec:geometry-rewrite}. 
In Section \ref{sec:rewrite-homogenization}, Proposition \ref{prop:Dbar-has affine-coordinates} and Theorem
\ref{thm:rewrite-square-aps} we saw that the Puiseaux branches in $\overline{S_t}$
on $\overline{D}$ that specialize to the points on $\overline{S_0}$ and project to a point in $Y_0$
under $p$ correspond to the affine ($z_0\ne0\in \K$) 
solutions of the homogeneous system \eqref{eq:rewrite-homogeneous} over $\K$. 
As noted before, this is exactly the kind of problem 
studied in \cite[Section 4]{AGS26}.

So we let $Z\subseteq\K^N$ be
finite nonempty solution set of the system 
\eqref{eq:rewrite-square} viewed over $\K$, and 
let $y(z)$ denote the tuple $(z_{n+1},\ldots, z_{n+m})$ 
in the $n+m$-tuple $z=(z_1,\ldots,z_{n+m})$.
Proposition~\ref{prop:rewrite-lift} says that the
set of $y(z)(0)$,
as $z$ ranges over $Z$, are exactly $Y_0$.
Choose a linear form $\ell$ injective on $Y_0\cup\{0\}$.
Such forms exist because the excluded coefficient tuples lie in
finitely many proper hyperplanes. Then
\begin{equation}\label{eq:rewrite-fiber-polynomial}
  Q_{Y_0,\ell}(T)=\prod_{y\in Y_0}(T-\ell(y)),\qquad
 Q_{Y_0,\ell}(0)=0\ \Longleftrightarrow\ 0\in Y.
\end{equation}
We compute a polynomial with the same set of roots as $Q_{Y_0,\ell}$, allowing positive
multiplicities, using only the explicit homogeneous equations $G_i$.

\subsection{The perturbation of Ierardi and $H_{T,w}$}
\label{sec:rewrite-resultant}

Let $d_i=\deg(G_i)$, and $d$
be a degree bound on $d_i$, and let
\[
 \widehat G_i\coloneqq G_i+\varepsilon z_i^{d_i},\qquad
 B\coloneqq \prod_{i=1}^N d_i\leq d^N.
\]
Here $\varepsilon$ is independent of the line parameter $t$.
Let
\begin{equation}\label{eq:rewrite-marker}
 H_{T,w}\coloneqq Tz_0^2-z_0\ell(y(z))+\sum_{j=1}^N w^jz_j^2,
\end{equation}
and define
\begin{equation}\label{eq:rewrite-resultant}
 \mathscr R(T,t,w,\varepsilon)\coloneqq 
 \Res_{z_0,\ldots,z_N}
       (\widehat G_1,\ldots,\widehat G_N,H_{T,w}).
\end{equation}
The quadratic form $H_{T,w}$ over $\C(T,w)$ is chosen 
to output the factors of the polynomial 
$Q_{Y_0,\ell}(T)$. 
Similar to \cite[Lemma 4.8]{AGS26}, we note that 
if $z=(z_0:\ldots:z_N)\in Z$, then 
$\TT_w(H_{T,w}(z))=T-\ell(y(z))$ when $z_0=1$, and 
$H_{T,w}(z)=\sum\limits_{j=1}^N w^jz_j^2$ when $z_0=0$.
Here recall from Section \ref{sec:prelim-approximation}
that $\TT_v(P)$ is the trailing coefficient of $P$ 
with respect to a variable $v$.
So $H_{T,w}(z)$ is a constant free polynomial in $w$ 
exactly when $z_0=0\in \K$. This is what we
use to single out the affine solutions of the homogeneous
system \eqref{eq:rewrite-homogeneous}.
Let $Z^*_0\subset \PPj^N_\K$ be a finite set of limit points in the sense of \cite[Definition 4.5 and Lemma 4.6]{AGS26} for the system defined by polynomials $\widehat{G}_i$ over $\C(t,\varepsilon)$.
Then $Z^*_0$ contains all the isolated projective points over $\K$ of the system \eqref{eq:rewrite-homogeneous} and is finite.
\begin{remark}
    The set $Z^*_0$ contains may contain a finite set of points from the positive dimensional components of $V(G_1,\ldots, G_N)$ over $\K$ at $z_0=0$. However, the affine points in $Z^*_0$ 
    are exactly the ones in $Z$, i.e., the solutions to the affine system \eqref{eq:rewrite-square} over $\K$, as the system \eqref{eq:rewrite-square} defines the 0-dimensional $Z$ over $\K$.
\end{remark}
We now state the standard factorization lemma of 
\cite[Lemma 2.6]{Ier89} applied to our polynomials 
$\widehat{G_i}$ and $H_{T,w}$.
\begin{lemma}[{\cite[Lemma 2.6]{Ier89} or \cite[Lemma 4.7]{AGS26}}]\label{lem:Ierardi}
  Let $G_i$, $\widehat{G_i}$, $H_{T,w}$, and $Z_0^*$ be as above.
     For each $z\in Z^*_0$, there are positive integers 
  $\mu_z>0$ such that 
  \[  (\TT_\varepsilon\mathscr R)(T,t,w)=\prod\limits_{z\in Z^*_0}H_{T,w}(z)^{\mu_z}.
  \]
\end{lemma}
Recall that each $z\in Z$, corresponds to the projective
solution $[1:z]$ to the system 
\eqref{eq:rewrite-homogeneous}, so we choose $z_0=1$. As a consequence of the above lemma we obtain that
\begin{equation}\label{eq:trail-with-w}
  \TT_w\left(\prod\limits_{z\in Z^*_0}H_{T,w}(z)^{\mu_z}\right)=c(t)\prod\limits_{z\in Z}H_{T,w}(z)^{\mu_z}=c(t)\prod\limits_{z\in Z}(T-\ell(y(z)))^{\mu_z},
\end{equation}
for some $c(t)\in \mathcal O_\K$. Here 
$c(t)$ turns out to be an element in $\mathcal O_\K$ 
precisely because it is coming from the points $z\in Z^*$
with $z_0=0$, and for such points we can clear the 
denominators and assume that $z_1,\ldots, z_N\in 
\mathcal O_\K$. Futhermore, at least one of $z_1,\ldots,z_N$ is non-zero, so $c(t)\ne 0$. 

\subsection{$R_\ell$ and $Q_{Y_0,\ell}$ have same roots}\label{sec:same-root}
In this section we finally show that the univariate polynomials 
$R_\ell$ and $Q_{Y_0,\ell}$ have same roots as claimed.
\begin{theorem}
\label{thm:rewrite-achar}
Let $R_\ell$, and $\mathscr R$ be as above. Then there is a $c\in \C^\times$ such that  
\begin{equation}\label{eq:rewrite-achar}
 R_\ell(T):=\TT_t\TT_w\TT_\varepsilon\mathscr R
   =c\prod_{y\in Y_0}(T-\ell(y))^{M_y},
 \quad M_y>0.
\end{equation}
Therefore,
\begin{equation}\label{eq:rewrite-decision}
 F\text{ has APS}\quad\Longleftrightarrow\quad R_\ell(0)=0.
\end{equation}
\end{theorem}
\begin{proof}
By Proposition~\ref{prop:rewrite-lift}, for $z\in Z$
the coordinates of $y(z)$ are elements of 
$\mathcal O_\K$, and so have
non-negative valuation with respect to $\val_t$. Therefore, 
$\val_t(\ell(y(z)))\geq 0$, for $z\in Z$. 
By Lemma \ref{lem:Ierardi} and equation \eqref{eq:trail-with-w}, we obtain 
\[R_\ell(T)=\TT_t\left(c(t)\prod\limits_{z\in Z}(T-\ell(y(z)))^{\mu_z}\right)=\TT_t(c(t))\prod\limits_{z\in Z}(T-\ell(y(z)(0)))^{\mu_z}.\]
The last equality above follows because at the factors
$\TT_t(T-\ell(y(z)))=T-\ell(y(z)(0))$, and the 
trailing coefficient operation distributes over a
product. Once again by Proposition 
\ref{prop:rewrite-lift}, the set 
$\{y(z)(0)\ |\ z\in Z\}$ is the same as $Y_0$, so we obtain 
\[R_\ell(T)=c\prod\limits_{y\in Y_0}(T-\ell(y))^{M_y},\]
where $c\in \C^\times$ and $M_y>0$. Note that for 
$y(z)(0)=y$, $M_y$ may not be equal to $\mu_z$, since we only have a equality of sets 
$\{y(z)(0)\ |\ z\in Z\}=Y_0$. This finishes the proof.
\end{proof}
In Example~\ref{ex:line-slice},
the polynomials $G_i=x_0^{\deg(g_i)}
g_i(x_1/x_0,\ldots,y_3/x_0)$, where 
$g_i$ polynomials over $\C[t]$ given in
\eqref{eq:overview-example-square}. 
The solution set over $K$ consists of
$p(t)=(t:t^2:t+1:t^2:t^2:t(t^2+t-1))$ and 
$[0:0:a:0:0:b]\simeq \PPj^1_\K$. The operation
$\TT_\varepsilon$ will get rid of all but finitely many points of the component 
$\PPj^1_\K$ at $x_0=0$, 
and the only isolated affine point is $p(t)\in \PPj^{5}_\K$. Choose $\ell(y)=y_3$. Over $\C(t)$, the affine system has the unique solution
\[
(x_1,x_2,y_1,y_2,y_3)
=
\left(t,\frac{t+1}{t},t,t,t^2+t-1\right).
\]
This solution is simple: the equations determine $y_1,y_2,x_1,x_2,y_3$ successively, with the coefficient of $x_2$ equal to $-t\ne0$.

For the projective representative $p(t)$ displayed above,
\[
\left.H_{T,w}(p(t))\right|_{w=0}
=
t^2\bigl(T-(t^2+t-1)\bigr).
\]
Every point of $Z_0^*$ at infinity contributes a factor whose trailing coefficient in $w$ is nonzero and independent of $T$. Since the unique affine solution is simple, the resultant factorization therefore gives
\[
\TT_w\TT_\varepsilon\mathscr R
=
c(t)\bigl(T-(t^2+t-1)\bigr)
\]
for some nonzero $c(t)\in\C[t]$. Here $c(t)$ is a polynomial because it is the coefficient of $T$ on the left-hand side.

Taking the trailing coefficient in $t$, we obtain
\[
R_\ell(T)=c(T+1),
\qquad
c=\TT_t(c(t))\ne0.
\]
Thus $R_\ell$ has the same roots as $Q_{Y_0,\ell}(T)=T+1$. In particular, $R_\ell(0)\ne0$, so Theorem~\ref{thm:rewrite-achar} confirms that the original system has no APS.

\subsection{The test in the Counting Hierarchy}
\label{sec:rewrite-ch}

Theorem~\ref{thm:rewrite-achar} reduces APS, for generic choices
of the auxiliary parameters, to the test
\[
  \left(\TT_t\TT_w\TT_\varepsilon\mathscr R\right)(0)=0.
\]
To obtain a Counting Hierarchy algorithm, we must implement this
test from the given arithmetic circuits. This requires choosing
the auxiliary parameters with polynomial bit length and bounded
failure probability, and accessing the required coefficients
without expanding the resultant.

Section~\ref{sec:effective} supplies these details. In particular,
the resultant degrees and coefficient bit lengths are at most
singly exponential in the input length, so monomial exponents
and bit positions have polynomial-length encodings. Access to
individual coefficient bits, together with a fixed number of
quantifier blocks over these indices, suffices to perform the
three trailing-coefficient operations and the final vanishing
test.

\section{Effective construction in the Counting Hierarchy}
\label{sec:effective}

We now implement the criterion in
\cref{eq:rewrite-decision} from the given  arithmetic
circuits. Throughout this section, $L$ denotes the original input
length, including the binary encodings of rational constants.
All circuits constructed below have size polynomial in $L$.
Their formal degrees, and the bit lengths of their expanded
coefficients, may be singly exponential.

\subsection{Dimension and reduction of the input}

We first determine the image dimension needed for
Proposition~\ref{prop:rewrite-normal-form}.

\begin{lemma}[Image dimension from exact feasibility]
\label{lem:dimension}
For a polynomial map given by  rational circuits,
the dimension of its image closure is computable in deterministic
polynomial time with an oracle at a fixed level of $\CH$.
\end{lemma}

\begin{proof}
Let $F=(f_1,\ldots,f_m)$ be a map in $n$ variables, and let
$J_F$ be its Jacobian matrix. In characteristic zero, the
transcendence degree of $\Q(f_1,\ldots,f_m)$ equals the rank of
$J_F$ over $\Q(x_1,\ldots,x_n)$.

For $1\leq k\leq\min\{m,n\}$, introduce matrices of variables
$B\in\A^{k\times m}$ and $C\in\A^{n\times k}$, and a variable $v$.
Then
\begin{equation}\label{eq:rank-hn}
  \rank_{\Q(x)}J_F\geq k
  \quad\Longleftrightarrow\quad
  v\det(BJ_F(x)C)-1=0
  \text{ has a complex solution}.
\end{equation}
Indeed, a nonzero $k$-minor can be made nonzero at some point
$x$, and $B,C$ can select that minor. Conversely, an invertible
$BJ_F(x)C$ certifies rank at least $k$.

Differentiation and the determinant construction
of~\cite{Ber84} give a polynomial-size  circuit
for the polynomial in~\eqref{eq:rank-hn}. Introduce one variable
for each gate value, impose its defining equation, and set the
output to zero. This gives a polynomial-size quadratic HN
instance with the same feasibility status. By
\cite[Theorem~6.9]{AGS26}, its feasibility is decidable at a fixed
level of $\CH$. Testing the polynomially many values of $k$
determines the image dimension.
\end{proof}

Write $r$ for this dimension. If $r=0$, the image consists of
one point, so APS is equivalent to exact feasibility of the
original system. We decide this using HN. If $r=m$, the image
closure is $\A^m$, so APS holds.

For $1\leq r<m$, apply
Proposition~\ref{prop:rewrite-normal-form}, choosing its sampling grid so that the failure probability is at most $1/12$. The logarithms of the degree bounds in that proposition are polynomial in $L$, so the sampled integers have polynomial bit length. The affine subspace can be parametrized by rational linear algebra, and substitution preserves polynomial circuit size. If the sampled equations do not define an $r$-dimensional affine subspace, we return an arbitrary answer; this is included among the unsuccessful reduction choices.

Without loss of generality, we have the polynomial map $F=(f_1,\ldots,f_m)$ such that $m=n+1$, and $\dim\overline{F(\A^n)}=n\geq1$ for which we need to check the existence of APS. Let $d\geq2$ bound its formal degrees. Then $d\leq2^{\poly(L)}$. Use the equations and homogenizations from \cref{eq:rewrite-square,eq:rewrite-homogeneous}, with
\[
  N=n+m,\qquad
  1\leq d_i\leq d,\qquad
  B=\prod_{i=1}^N d_i\leq d^N.
\]
Here the $d_i$ are fixed homogenizing degrees, chosen before
sampling the remaining parameters. Homogenization can be
performed gate by gate using formal degrees; powers of the
homogenizing variable are computed by repeated squaring.
Consequently, the homogeneous equations and their perturbations
still have polynomial-size circuit descriptions.

\subsection{Choosing the auxiliary parameters}

Let $\Theta$ collect the coefficients of the projection $\pi$,
the direction $\alpha$, and the linear form $\ell$. There are
polynomially many such parameters. For the fixed reduced input,
regard them first as independent indeterminates, and form the
resultant from~\eqref{eq:rewrite-resultant}:
\[
  \mathscr R_\Theta
  \in\Q[\Theta,\varepsilon,t,w,T].
\]
The geometric argument of
Theorem~\ref{thm:rewrite-achar} applies over an algebraic closure
of $\Q(\Theta)$. Hence, defining
\begin{equation}\label{eq:universal}
  Q_\Theta(T)
    :=\TT_t\TT_w\TT_\varepsilon\mathscr R_\Theta,
  \qquad
  F\text{ has APS}
    \quad\Longleftrightarrow\quad
  Q_\Theta(0)=0
  \text{ in }\Q[\Theta],
\end{equation}
we obtain a polynomial identity on yes-instances and a nonzero
polynomial on no-instances. We use this symbolic construction
only to analyze the random choices; the algorithm does not
construct $Q_\Theta$.

We next bound the degree of $\mathscr R_\Theta$. The resultant
is homogeneous of degree $2B/d_i$ in the coefficients of the
$i$th perturbed equation and degree $B$ in the coefficients of
$H_{T,w}$. Its total degree in these coefficient groups is
therefore at most $(2N+1)B$.
Each substituted coefficient has total degree at most $N+1$
in $\Theta,\varepsilon,t,w,T$. Thus
\begin{equation}\label{eq:Dstar}
  \deg\mathscr R_\Theta
  \leq
  D_*:=(N+1)(2N+1)d^N
  =2^{\poly(L)}.
\end{equation}
In particular, $D_*$ has a polynomial-length binary encoding.

The following lemma controls the effect of substituting
numerical values for $\Theta$.

\begin{lemma}[Stability of nested trailing coefficients]
\label{lem:stability}
Let
\[
  R(\Theta,t_1,\ldots,t_k,T)
  \in\Q[\Theta,t_1,\ldots,t_k,T]
\]
be nonzero and have total degree at most $D$.
There are nonzero polynomials
$A_1,\ldots,A_k\in\Q[\Theta]$, each of degree at most $D$,
such that, whenever $A_i(\theta)\neq0$ for all $i$,
\[
  \left(\TT_{t_k}\cdots\TT_{t_1}R\right)(\theta,T)
  =
  \TT_{t_k}\cdots\TT_{t_1}
  \bigl(R(\theta,t_1,\ldots,t_k,T)\bigr).
\]
\end{lemma}

\begin{proof}
Take the first nonzero coefficient of $R$ in $t_1$.
Choose a nonzero coefficient $A_1(\Theta)$ of one of its
monomials in the remaining variables. If $A_1(\theta)\neq0$,
this coefficient polynomial remains nonzero after substitution.
All lower coefficients were identically zero, so the first
trailing exponent is unchanged.

Repeat this argument on the selected coefficient polynomial,
successively for $t_2,\ldots,t_k$. Extracting coefficients never
increases total degree, so each $A_i$ has degree at most $D$.
\end{proof}

Let $M=2^k$ for some $k$ such that $M\geq48D_*$, and sample each coordinate of $\theta$ independently and uniformly from $\{1,\ldots,M\}$.
Apply Lemma~\ref{lem:stability} with
\[
  (t_1,t_2,t_3)=(\varepsilon,w,t).
\]
If none of its three exceptional polynomials vanishes, then
\[
  \TT_t\TT_w\TT_\varepsilon\mathscr R_\theta
  =Q_\Theta(\theta,T).
\]
On a yes-instance, the constant coefficient on the right is
zero by~\eqref{eq:universal}. On a no-instance,
$Q_\Theta(0)$ is a nonzero polynomial in $\Theta$ of degree at
most $D_*$, and its specialization can vanish only on one
additional exceptional set. Schwartz--Zippel therefore gives
\begin{equation}\label{eq:failure}
  \Pr_\theta\!\left[
    \left(\TT_t\TT_w\TT_\varepsilon\mathscr R_\theta\right)(0)=0
    \text{ gives the wrong answer}
  \right]
  \leq \frac{4D_*}{M}
  \leq \frac1{12}.
\end{equation}
Every sampled integer has polynomial bit length. Neither the
exceptional polynomials nor $Q_\Theta$ need to be computed.

\begin{remark}[Nonzero resultant for every sample]
\label{rem:nonzero}
The polynomial $\mathscr R_\theta$ is nonzero for every
specialization $\theta$, including unsuccessful choices.
Indeed, its coefficient of the highest possible power of
$\varepsilon$ is
\[
  \Res_{z_0,\ldots,z_N}
  \bigl(z_1^{d_1},\ldots,z_N^{d_N},H_{T,w}\bigr)
  =\pm T^B.
\]
Thus all three trailing-coefficient operations are defined for
every sample.
\end{remark}

\subsection{Coefficient access and the final test}

We now fix a sampled tuple $\theta$ of polynomial bit length.
The remaining auxiliary variables are
$\varepsilon,t,w,T$.

Since the input coefficients are rational, we allow a
coefficient to be represented as $a/b$, where
$a\in\Z$ and $b\in\Z_{>0}$; this representation need not be
reduced. Access to its bits means access to the sign and binary
magnitude bits of $a$ and to the bits of $b$.
A coefficient vanishes exactly when its numerator vanishes.

\begin{theorem}[Resultant coefficient access]
\label{thm:ags}
For the resultant $\mathscr R_\theta$, a specified bit of the
numerator or denominator representing a specified coefficient
is computable at a fixed level of $\CH$.
Its degrees and the bit lengths of these representations are
bounded by $2^{\poly(L)}$.
These statements hold uniformly over all sampled tuples
$\theta$ of polynomial bit length.
\end{theorem}

\begin{proof}
The coefficient-access argument in
Section~\ref{sec:prelim} provides a $\CH$ oracle for the
coefficients of the input forms. Their implicit dense
descriptions have singly exponential length.

By \cite[Corollary~5.13]{AGS26}, the resultant coefficients are
computed by uniform constant-depth arithmetic circuits of size
\[
  \bigl(P(N+1)\bigr)^{\poly(N)}
  =2^{\poly(L)},
  \qquad
  P=B+\sum_{i=1}^N\frac{2B}{d_i}.
\]
Here there are four auxiliary variables, and their degrees in
the input coefficients are at most $N+1$.
The arithmetic-to-threshold simulation of
\cite[Lemma~6.1]{AGS26} preserves constant depth and gives
singly exponential size and output bit-length bounds.

The gate-evaluation argument of \cite[Lemma~6.3]{AGS26} applies
with the input bits supplied by the coefficient oracle:
gate indices have polynomial length, input gates are evaluated
in $\CH$, and each further threshold layer adds one majority
quantifier. Since the depth is constant, output bits remain at
a fixed level of $\CH$. The degree bound follows from
\eqref{eq:Dstar}.
\end{proof}

In particular, exponentially long coefficients cause no
difficulty: a position within such a coefficient still has a
polynomial-length binary encoding. The next lemma makes the
resulting vanishing test explicit.

\begin{lemma}[Trailing-coefficient test]
\label{lem:trailing-ch}
Let $R(t_1,\ldots,t_k,T)$ be a nonzero rational polynomial,
where $k$ is fixed.
Suppose its degrees and the numerator and denominator bit
lengths in its coefficient representations are bounded by
$2^{\poly(L)}$, and that individual coefficient bits are
accessible at a fixed level of $\CH$.
Then the predicate
\[
  \left(\TT_{t_k}\cdots\TT_{t_1}R\right)(0)=0
\]
also belongs to a fixed level of $\CH$.
\end{lemma}

\begin{proof}
Write
\[
  R=\sum_{e,j}c_{e,j}t_1^{e_1}\cdots t_k^{e_k}T^j.
\]
Choose a bound $U=2^{p(L)}$ large enough to cover all exponents
and coefficient bit lengths. Each exponent and bit position
then has a polynomial-length encoding.

If $c_{e,j}=a_{e,j}/b_{e,j}$ with $b_{e,j}>0$, then
\[
  c_{e,j}\neq0
  \quad\Longleftrightarrow\quad
  \exists\,q<U:
  \text{the $q$th magnitude bit of $a_{e,j}$ is $1$}.
\]
Thus coefficient vanishing and nonvanishing are decidable in
$\CH$.

The successive trailing operations select the lexicographically
least tuple $e$ for which some $c_{e,j}$ is nonzero, with
$e_1$ compared first. Guess this tuple and require
\[
  \exists j\;c_{e,j}\neq0,
  \qquad
  \forall e'<_{\mathrm{lex}}e\;\forall j\;c_{e',j}=0.
\]
Finally, require $c_{e,0}=0$.
All indices are bounded by $U$, and all quantified strings have
polynomial length. Expanding the coefficient tests introduces
only a fixed number of additional quantifier blocks over the
coefficient-bit oracle. Hence the whole predicate lies at a
fixed level of $\CH$.
\end{proof}

\begin{proof}[Proof of Theorem~\ref{thm:aps}]
Compute the image dimension using Lemma~\ref{lem:dimension}.
Handle the cases $r=0$ and $r=m$ as described above.
Otherwise, perform the reduction of
Proposition~\ref{prop:rewrite-normal-form}, with failure
probability at most $1/12$.

For the reduced map, choose the auxiliary parameters from the
grid specified above and form the circuit descriptions of the
homogeneous equations and their perturbations.
By Theorem~\ref{thm:ags} and Lemma~\ref{lem:trailing-ch}, we can
test
\[
  \left(\TT_t\TT_w\TT_\varepsilon
  \mathscr R_\theta\right)(0)=0
\]
using an oracle at a fixed level of $\CH$. This test is defined
for every sample by Remark~\ref{rem:nonzero}.

Conditional on a successful input reduction,
Theorem~\ref{thm:rewrite-achar} and~\eqref{eq:failure} show that
the test is correct with probability at least $11/12$.
Including the input reduction, the total error probability is
at most $1/6$.

All random choices use polynomially many bits, and all remaining
computations take polynomial time with the stated oracles.
Choose a fixed $b$ large enough to include these oracles in
$C_b\mathsf P$. The resulting algorithm belongs to
\[
  \mathsf{BPP}^{C_b\mathsf P}
  \subseteq
  \mathsf{PP}^{C_b\mathsf P}
  =
  C_{b+1}\mathsf P
  \subseteq\CH.
\]
This proves the theorem for  rational circuit
input. Sparse input with binary exponents converts to this
representation in polynomial time by repeated squaring.
\end{proof}

\begin{corollary}\label{cor:nonexact}
For  rational circuit input, and hence for sparse
input with binary rational coefficients and binary exponents,
deciding whether a system has an approximate zero but no exact
zero belongs to $\CH$.
\end{corollary}

\begin{proof}
Combine Theorem~\ref{thm:aps} with HN in $\CH$ and closure of
$\CH$ under complement and finite Boolean combinations.
\end{proof}

\begin{remark}[Other coefficient fields]\label{rem:APS-for-other-fields}
Theorem~\ref{thm:aps} also holds for arithmetic circuits over
number fields and finite fields, with image closure taken over the algebraic closure of the coefficient field. We use the standard representation by an irreducible defining polynomial and power-basis coordinates; for finite fields, the characteristic is given in binary. We need the following modifications to accommodate positive characteristic \begin{itemize}
    \item The image dimension cannot, in general, be computed by Jacobian rank. Instead, for a tuple of algebraically independent indeterminates $z$, we compute the dimension of the generic fiber $F(x)=F(z)$ over $k(z)$ and subtract it from the number of input variables. This computation is in $\CH$ by \cite[Corollary~7.2]{AGS26}, using auxiliary variables for the arithmetic circuit gates.
    \item Writing $K=\overline{k}$, we replace Puiseux series by elements of $\overline{K((t))}$ equipped with the extended $t$-adic valuation, and interpret limits as residues of elements of nonnegative valuation. The lifting argument uses normalization of curves and extension of valuations, allowing inseparable extensions as well; see \cite[Tag~0ASF]{Stacks26}. Noether normalization and the resulting integrality argument remain valid. Products over solutions retain their algebraic multiplicities, which do not affect whether the specialized constant coefficient vanishes.
    \item Finally, generic parameters are sampled from a sufficiently large finite extension of $k$. An extension of polynomial degree suffices for our degree bounds and can be constructed using \cite[Lemma~2.13]{AGS26}. The resultant coefficient tests use the finite-field version of \cite[Theorem~6.7]{AGS26}, preserving the $\CH$ upper bound. No lower bound on the characteristic relative to the input degrees is required. For number fields, the characteristic-zero geometric argument applies, with arithmetic performed in the given number-field representation.
\end{itemize}
\end{remark}

\section{Integral closure membership in the Counting Hierarchy}\label{sec:integral}
Here, we prove \cref{thm:icm}. We do it by reducing integral closure membership testing to $\APS$. 

The well-known valuative criterion for integral dependence~\cite[Corollary 5.22]{AM69} establishes a connection between integrality and valuation theory, see Corollary~\ref{cor:boundedness} for an equivalent version useful for our purposes. Meanwhile there also exists a connection between valuation theory and APS, see \cref{thm:WAN}. This suggests an underlying connection between integrality and APS, which was also made somewhat explicit in the approximate Rabinowitsch's trick in \cite[Lemma 2]{ADPS26}. In the rest of this section, we formalize this intuition to get the desired criterion.

Fix $f,f_1,\ldots,f_m\in\C[x_1,\ldots,x_n]$, and let $A=\C[f_1,\ldots,f_m]$. Without loss of generality, we can assume $f\neq0$, otherwise it is always integral. Using \cref{thm:WAN}, it is easy to see the following connection between APS of a locally defined system of polynomials and their valuations.

\begin{lemma}\label{lem:aps-valuation-connection}
    Given $c=(c_1,\ldots, c_m)\in \C^m$ consider the system of polynomials $\Phi_c := (f_1-c_1,\ldots,f_m-c_m,u,1-uf)$ where each polynomial lies in $\C[x_1,\ldots,x_n,u]$. Then, $\Phi_c$ has an APS if and only if there exists some $\beta \in \C(t)^{n+1}$, such that for each polynomial $h \in \Phi_c$, $\val_t(h(\beta)) > 0$.
\end{lemma}
We can now use the connection of APS to annihilators to get an equivalent local ``annihilator'' for this locally defined system of polynomials.
\begin{lemma}\label{lem:empty-pointwise-fiber-criterion}
Fix $c=(c_1,\dots,c_m)\in \C^m$. Then, the system of polynomials $\Phi_c = (f_1-c_1,\dots,f_m-c_m,u,1-uf)$ has no APS if and only if there exists
\[
Q_c(Y_1,\dots,Y_m,U,V)\in \C[Y_1,\dots,Y_m,U,V]
\]
such that \begin{enumerate}
    \item $Q_c(f_1,\dots,f_m,u,1-uf)=0$, and
    \item $Q_c(c_1,\dots,c_m,0,0)=1$
\end{enumerate}
\end{lemma}

\begin{proof}
By the criterion of \cite{GSS19,ADPS26}, non-existence of APS for $\Phi_c$ is equivalent to existence of a
polynomial $B_c(Z_1,\dots,Z_m,U,V)$ with constant term equal to $1$ after rescaling, such that $B_c(f_1-c_1,\dots,f_m-c_m,u,1-uf)=0$. Now define
\[
Q_c(Y_1,\dots,Y_m,U,V):=
B_c(Y_1-c_1,\dots,Y_m-c_m,U,V).
\]
Then $Q_c(f_1,\dots,f_m,u,1-uf)=0$, and $Q_c(c_1,\dots,c_m,0,0)=B_c(0,\dots,0,0,0)=1$.

For the converse, if a $Q_c$ satisfying $(i)$ and $(ii)$ exists, then define
\[
B_c(Y_1,\dots,Y_m,U,V):=
Q_c(Y_1+c_1,\dots,Y_m+c_m,U,V).
\]
Clearly, $B_c(0,\ldots,0)=Q_c(c_1,\ldots,c_m,0,0) = 1$, and $B_c(f_1-c_1,\dots,f_m-c_m,u,1-uf)=Q_c(f_1,\dots,f_m,u,1-uf)= 0$. Hence, $B_c$ is the annihilator for $\Phi_c$ with non-zero constant term, and hence $\Phi_c$ has no APS. 
\end{proof}

We now come to our main observation: these ``local'' annihilator certificates $Q_c$ can in fact be converted into a global certificate $Q$. This is essential as we have infinitely many local certificates which makes any computational strategy impossible to execute. Let $\Phi$ denote the system of polynomials $f_1,\ldots,f_m,u,1-uf$, i.e., $\Phi:=(f_1,\ldots,f_m,u,1-uf)$, and $\Ann(\Phi)$ denote the annihilator ideal of polynomials in $\Phi$.

\begin{theorem}[Equivalence of the pointwise and global criteria]
\label{thm:pointwise-global-equivalence}
With the above notation, the following are equivalent:
\begin{enumerate}[label=(\arabic*)]
    \item For every \(c=(c_1,\dots,c_m)\in \C^m\), there exists a polynomial $Q_c(Y_1,\dots,Y_m,U,V)\in \Ann(\Phi)$ such that $Q_c(c_1,\dots,c_m,0,0)=1$.
    \item There exists a single polynomial $Q(Y_1,\dots,Y_m,U,V)\in \Ann(\Phi)$ such that $Q(Y_1,\dots,Y_m,0,0)=1$.
\end{enumerate}
\end{theorem}

\begin{proof}
Let $J:=\bigl(\Ann(\Phi)+(U,V)\bigr)\cap \C[Y_1,\dots,Y_m]$. Thus, the ideal $J$ may be viewed as the ideal of all \(Y\)-polynomials that arise from annihilator relations in $\Ann(\Phi)$ after forgetting all terms divisible by \(U\) or by \(V\). Assume first that (1) holds. Fix \(c\in \C^m\). By hypothesis there exists \(Q_c\in \Ann(\Phi)\) with \(Q_c(c,0,0)=1\). Passing to the quotient modulo \((U,V)\), the image
\(
q_c(Y):=Q_c(Y,0,0)
\)
lies in \(J\), and satisfies \(q_c(c)=1\neq 0\). Hence \(c\) is not a common zero of \(J\). Since this is true for every \(c\in \C^m\), the variety \(V(J)\) is empty. By Hilbert's Nullstellensatz, $J=\C[Y_1,\dots,Y_m]$, and \(1\in J\).  By definition of \(J\), \(1\in J\) means that there exists an annihilator \(Q\in \Ann(\Phi)\) whose restriction to the slice \(U=V=0\) is
the constant polynomial \(1\). That is, there exists
\(Q\in \Ann(\Phi)\) and \(R,S\in \C[Y_1,\dots,Y_m,U,V]\) such that $1=Q+UR+VS$. Setting \(U=V=0\) gives $Q(Y,0,0)=1$. Thus (2) holds.

Conversely, assume (2). Then for every \(c\in \C^m\), evaluating the identity \(Q(Y,0,0)=1\) at \(Y=c\) yields
\(Q(c,0,0)=1.\) Since \(Q\in \Ann(\Phi)\), the same polynomial \(Q\) serves as a witness in (1) for every \(c\). Thus (1) holds.
\end{proof}
\begin{remark}\label{rem:line-bundle-version}
From an algebraic geometric viewpoint one can also see Theorem \ref{thm:pointwise-global-equivalence} as an explicit version of a classical result: finitely generated locally trivial modules of rank 1 over a unique factorization domain are globally trivial. 
Indeed $\Ann(\Phi_0)$ is a locally finitely generated module of $R\coloneqq \C[Y_1,\ldots, Y_m, U,V]$. Our local conditions $(1)$ corresponds in Theorem \ref{thm:pointwise-global-equivalence} mean that at each $c\in \A^m$ the localized 
module 
\[\left(\Ann(\Phi_0)\otimes \frac{R}{(U,V)}\right)_{c}\simeq Q_c\cdot \left(\frac{R}{(U,V)}\right)_{c}\simeq \left(\frac{R}{(U,V)}\right)_{c},
\]
i.e., locally trivial of rank 1 over $(R/(U,V))_c$ via the witness $Q_c$.
Hence, there exists a $Q\in R$ such that 
\[\Ann(\Phi_0)\otimes \frac{R}{(U,V)}\simeq Q\cdot \frac{R}{(U,V)}\simeq \frac{R}{(U,V)}
\]
as $R/(U,V)$-modules. 
 \end{remark}
Thus, the above discussion relates non-existence of APS for all local systems to a global polynomial certificate. But, we still don't have an explicit connection to integrality. This is provided in the following theorem.
\begin{theorem}\label{thm:integrality-global-criterion}
A polynomial 
$f$ is integral over $A=\C[f_1,\ldots,f_m]$ if and only if there exists \(Q(Y_1,\dots,Y_m,U,V)\in \C[Y_1,\dots,Y_m,U,V]\) such that $Q(f_1,\dots,f_m,u,1-uf)=0$ and $Q(Y_1,\dots,Y_m,0,0)=1$.
\end{theorem}
\begin{proof}
    Suppose $f$ is integral over $A$. Then, there exists a monic relation $f^r+a_{r-1}(f_1,\dots,f_m)f^{r-1}
+\cdots+a_1(f_1,\dots,f_m)f+a_0(f_1,\dots,f_m)=0$ for some $r\in \Z_{> 0}$ and some $a_i(Y_1,\dots,Y_m)\in \C[Y_1,\dots,Y_m]$. Define 
\[
Q(Y_1,\dots,Y_m,U,V):=
(1-V)^r
+a_{r-1}(Y)\,U(1-V)^{r-1}
+\cdots+
a_1(Y)\,U^{r-1}(1-V)
+a_0(Y)\,U^r.
\]
Clearly, $Q(Y_1,\dots,Y_m,0,0)=(1-0)^r=1$. Further, on substituting $Y_i=f_i,U=u,V=1-uf.$ Since $1-V=uf$, we obtain
\begin{align*}
Q(f_1,\dots,f_m,u,1-uf)
&=
(uf)^r+a_{r-1}(f)\,u(uf)^{r-1}
+\cdots+a_0(f)\,u^r \\
&=
u^r\Bigl(f^r+a_{r-1}(f)f^{r-1}+\cdots+a_0(f)\Bigr) \\
&=0.
\end{align*}

For the converse, let $Q(Y,U,0)=1+b_1(Y)U+b_2(Y)U^2+\cdots+b_r(Y)U^r$ for some $r\geq 0$ and some polynomials $b_i(Y)\in \C[Y_1,\dots,Y_m]$. Since
\( Q(f_1,\dots,f_m,u,1-uf)=0 \) in $\C[x_1,\dots,x_n,u]$, the same identity holds in the localization $\C(x_1,\dots,x_n)[u]$. Now substitute $u=\frac{1}{f}$ in the fraction field $\C(x_1,\dots,x_n)$. Then $1-uf=0$. Hence, 
\[
0 =
Q\!\left(f_1,\dots,f_m,\frac1f,0\right) =
1+\frac{b_1(f_1,\dots,f_m)}{f}
+\cdots+
\frac{b_r(f_1,\dots,f_m)}{f^r}.
\]
Multiplying by $f^r$ gives a monic polynomial relation for $f$ over $A$.
\end{proof}
In fact, this can also be extended to an equivalence with valuations.
\begin{corollary}[Boundedness criterion]\label{cor:boundedness}
The polynomial $f$ is integral over $A$ if and only if, for every $\beta\in\mathcal \C(t)^n$,
\[
 \bigl(\val_t f_i(\beta)\geq0\text{ for all }i\bigr)
 \quad\Longrightarrow\quad \val_t f(\beta)\geq0.
\]
\end{corollary}
\begin{proof}
For a monic relation, a negative valuation of $f(\beta)$ makes its highest power the unique term of least valuation, a contradiction. Conversely, if
$f$ is not integral, then Lemma~\ref{lem:empty-pointwise-fiber-criterion} and \cref{thm:pointwise-global-equivalence,thm:integrality-global-criterion} imply that there is an approximate zero $(\beta,u(t))$ of a $\Phi_c$ for some $c\in \C^m$. All $f_i(\beta)$ are nonnegative valuations, whereas $u$ has positive valuation and $ \val_t(1-u(t))f(\beta)>0$. Hence $\val_t(f(\beta))=-\val_t(u(t))<0$.
\end{proof}

Hence, we can now check for integrality by checking for this global annihilator certificate. However, computationally the issue is that we currently have no handle on the degree of this polynomial $Q$, which could make computation very expensive. But as we now show, we can establish a singly exponential bound on this polynomial using Noether normalization. 

\subsection{A singly exponential certificate}

First note that we can always compute the transcendence degree of the input polynomials with rational coefficients in $\CH$ by just creating a minimally independent set by sequentially adding elements and checking for algebraic independence using the Jacobian criterion~\cite{chen2011partial}. Now, look at the system $\Phi=(f_1,\ldots,f_m,u,1-uf)$. If $f$ is integral over $A$, then $\trdeg \Phi = \rho + 1$, if $\trdeg (f_1,...,f_m) = \rho$. 
We use Noether normalization to project these $f_i$ to a space defined by their linear combinations $\{g_i\}_{1\leq i\leq \rho}$ such that $f$ is integral over this new system. Along with the polynomials $u,1-uf$, this makes a new system of $\rho+2$ polynomials whose transcendence degree is $\rho+1$. Hence, the annihilating ideal of the new system is a principal ideal~\cite{Kay09}. This helps us to establish the singly exponential degree bound on the annihilator.

\begin{theorem}\label{thm:exp-deg-bd-annihilator}
Assume $f$ is integral over $A=\C[f_1,\dots,f_m]$. Let $\rho:=\trdeg_{\C}\C(f_1,\dots,f_m)$, and $\Delta:=\max\{\deg(f_1),\dots,\deg(f_m),\deg(f)+1\}$. Then there exists $Q$ satisfying conditions in \cref{thm:integrality-global-criterion}, and $\deg(Q)\leq D:=(\rho+2)\Delta^{\rho+1}$. For rational input polynomials, $Q$ may be chosen over $\Q$.
\end{theorem}
\begin{proof}
    Let $f$ be integral over $A$. Since $\trdeg f_1,\ldots,f_m = \rho$, by applying first random linear projections Theorem \ref{thm:rewrite-random-projection} and then linear Noether 
    normalization Lemma \ref{lem:rewrite-linear-nn}, there exists $g_1,\ldots,g_{\rho}$ such 
    that $g_i = \sum_{j=1}^m c_{ij}f_j$ for $i=1,\ldots,\rho$. Let $\widetilde \Phi := (g_1,\ldots,g_{\rho},u,1-uf)$. For generic 
    choices of $c_{ij}$, the above procedure gives a Noether normalization of closure of the image of 
    $(f_1,\ldots,f_m):\A^n_\C\to \A^m_\C$,
    so $g_i$ are algebraically independent. Hence the ring $A$ is integral over $B:=\C[g_1,\ldots,g_{\rho}]$. Since, $f$ is integral over $A$, so $f$ is integral over $B$, and the transcendence degree of the system $\widetilde \Phi$ is $\rho + 1$. Thus, $\Ann(\widetilde \Phi)$ is a principal ideal. Then, using Theorem~\ref{thm:integrality-global-criterion} we have $P \in$ $\Ann(\widetilde \Phi)$ such that $P(Z,0,0)=1$ and by \cite[Theorem 11]{Kay09}, we have that $\deg(P) \leq (\rho+2)\Delta^{\rho+1}=D_0$. Now define \[ Q(Y,U,V) = P\left(
    \sum_j c_{1j}Y_j,\ldots,
    \sum_j c_{\rho j}Y_j,U,V
    \right).
    \]
    Clearly, $Q(Y,0,0) = P(Z,0,0) = 1$. Further, $Q(f_1,\ldots,f_m,U,V) = P(g_1,\ldots,g_{\rho},U,V) = 0$. hence, $Q \in \Ann(\Phi)$. Since, $Q$ only uses linear combinations of variables of $P$, $\deg Q = \deg P \leq D_0$.

\end{proof}

At this point, we can use this singly exponential degree bound to put integral closure testing in $\PSPACE$ by solving an exponential-sized system of linear equations
in the same way as for computing annihilator using Perron’s bound~\cite{Perron27,GSS19}, which can be solved in $\PSPACE$~\cite{csanky1976fast,Ber84,mulmuley1986fast}. However, this still does not give a $\CH$ upper bound, because while similar, this global annihilator certificate does not exactly correspond to any $\APS$ instance.

\subsection{A polynomial-time reduction to APS}

In this section, we show how we can get an equivalent $\APS$ instance for integral closure membership testing, which will finally put Problem~\ref{prob:int-cl} in $\CH$.

\begin{theorem}[From non-integrality to APS]\label{thm:power-reduction}
Let $D$ be the bound from \cref{thm:exp-deg-bd-annihilator}, or any larger integer, and let $M = D+1$. Define a system of polynomials $\Psi_M$ by $\Psi_M= (uf_1^M\ldots,uf_m^M,u,1-uf)$. Then $f$ is integral over $A$ if and only if $\Psi_M$ does not have APS.
\end{theorem}
\begin{proof}
If $f$ is not integral over $A$, Corollary~\ref{cor:boundedness} supplies $\beta\in\mathcal \C(t)^n$ with $\val_t f_i(\beta)\geq0$ for all $i$, and $ \val_t f(\beta)<0$. Substituting $u=1/f(\beta)$ we obtain $\val_t u(\beta)>0$, so $\val_t(u(\beta)f_i(\beta)^M)>0$
for all $i$.
Thus $(\beta,u(\beta))\in\C(t)^{n+1}$ is an approximate zero of $\Psi_M$ by \cite[Definition 1]{ADPS26}.

Conversely, suppose $f$ is integral and $(\beta,u)$ is an approximate zero of $\Psi_M$. Set $a=\val_t(u)>0$. Since $\val_t(1-uf(\beta))>0$, we obtain $\val_t(f(\beta))=-a$, and $\val_t(f_i(\beta))>-a/M$ for all $i$. The case $f=0$ is already excluded, since then $1-uf=1$. By \cref{thm:exp-deg-bd-annihilator}, there is a certificate $Q$ with $\deg Q\leq D$.
Write
\[
 Q(Y,U,0)=1+\sum_{j=1}^{r}b_j(Y)U^j,
 \qquad r\leq D,\quad \deg b_j\leq D-j.
\]
Substituting $u=1/f$ in its annihilating identity gives the rational-function identity
\begin{equation}\label{eq:reciprocal-certificate}
 1+\sum_{j=1}^{r}b_j(f_1,\ldots,f_m)f^{-j}=0.
\end{equation}
It can be evaluated at $\beta$ because $f(\beta)\ne0$. Each monomial $c_\alpha Y^\alpha$ in $b_j$ with nonzero evaluation contributes a term with valuation
\[
 ja+\sum_{i:\alpha_i>0}\alpha_i\val_t(f_i(\beta))
 \ \geq\ a\left(j-\frac{|\alpha|}{M}\right)
 \ \geq\ a\left(1-\frac{D}{M}\right)>0.
\]
A monomial evaluating to zero has infinite valuation and causes no exception. The coefficients $c_\alpha\in\C^\times$ have valuation zero. Hence every summand after $1$ in \cref{eq:reciprocal-certificate} has positive valuation and cannot cancel the constant term. This contradiction proves
the first equivalence. 
\end{proof}

\begin{proof}[Proof of \cref{thm:icm}]
Let $L$ be the circuit input length. Compute formal-degree upper bounds and choose an integer $\Delta\geq\max\{1,\deg f_1,\ldots,\deg f_m,1+\deg f\}$ from them. Define $D$ as in \cref{thm:exp-deg-bd-annihilator} with $\rho = n$ (as $\trdeg \leq n$ always) and $M=D+1$. These integers have polynomial binary length and can be computed in
polynomial time. The proof of \cref{thm:exp-deg-bd-annihilator} applies to these upper bounds as well as to the actual degrees. Repeated squaring computes each $f_i^M$ using $O(\log M)$ additional multiplication gates. Sharing the given circuits therefore yields a  rational circuit for $\Psi_M$ of size polynomial in $L$. 
By \cref{thm:power-reduction}, this is a deterministic polynomial-time many-one reduction from non-integrality to APS. Our APS algorithm from \cref{thm:aps} decides the resulting instance in $\CH$. Closure under complement of $\CH$ proves the claimed upper bound for integral closure membership.
\end{proof}

Thus, we have established a $\CH$ procedure of integral closure membership testing. The idea behind introducing $\Psi_M$ was to ensure that the valuation of $f_i$ dominate the valuation of $u$ and don't allow false witnesses of non-integrality by providing a uniform exponent without knowing a monic equation in advance. This can also be seen in the following example.

\begin{example}[Why the power is needed]\label{ex:power-needed}
Let $A=\C[x]$ and $f=x^k$ for $k\geq2$. The element $f$ belongs to $A$ and is therefore integral. Yet the tuple $(u,1-ux^k,ux)$ has an approximate zero: take $x=t^{-1}$ and $u=t^k$, giving $(t^k,0,t^{k-1})$. For the powered tuple $(u,1-ux^k,ux^M)$, approximate vanishing of its first two coordinates forces $\val_t(x)=-\val_t(u)/k$. Hence, if $M\geq k$, then
\[
 \val_t(ux^M)=\val_t(u)(1-M/k)\leq0,
\]
so this false witness is excluded. 
\end{example}

\begin{remark}\label{rem:icm-other-fields}
The global annihilator criterion and the reduction from
non-integrality to $\APS$ are valid in arbitrary characteristic. For the certificate degree bound, we use the
arbitrary-characteristic annihilator bound of
\cite[Corollary~6]{BeeckenMittmannSaxena13}. Using $\trdeg \leq n$ in this bound avoids any Jacobian computation. The same polynomial-time reduction therefore applies, and the claimed upper bound follows from the corresponding result for $\APS$.
\end{remark}

\begin{lemma}[From APS to non-integrality]
\label{lem:aps-to-nonintegrality}
Let $K$ be an algebraically closed field, let
$F=(f_1,\ldots,f_m)\in K[x_1,\ldots,x_n]^m$, and let $z$ be
a fresh variable. Then
\[
  0\in\overline{F(K^n)}
  \quad\Longleftrightarrow\quad
  z\text{ is not integral over }K[zf_1,\ldots,zf_m].
\]
This gives a deterministic polynomial-time many-one reduction
from APS to non-integrality.
\end{lemma}

\begin{proof}
If $F$ has no APS, the annihilator criterion gives a relation
\[
  H(Y)=1+\sum_{0<|\alpha|\leq D}h_\alpha Y^\alpha,
  \qquad H(F)=0.
\]
Multiplying this identity by $z^D$ gives
\[
  z^D+
  \sum_{0<|\alpha|\leq D}
  h_\alpha(zf_1)^{\alpha_1}\cdots(zf_m)^{\alpha_m}
  z^{D-|\alpha|}=0,
\]
a monic relation for $z$ over $K[zf_1,\ldots,zf_m]$.

Conversely, suppose that $z$ is integral over this algebra.
Write a monic relation as
\[
  z^d+\sum_{j=0}^{d-1}A_j(zf_1,\ldots,zf_m)z^j=0.
\]
Taking the coefficient of $z^d$ in the above monic relation produces an annihilator of $F$ with constant term $1$: every contribution from $A_j$ has positive total degree $d-j$ in the $f_i$. Thus $F$ has no APS.

The reduction simply introduces $z$ and multiplies each
input polynomial by it, so its circuit size is polynomial.
\end{proof}

\begin{proof}[Proof of \cref{thm:aps-icm-equal}] 
    Lemma~\ref{lem:aps-to-nonintegrality} along with \cref{thm:power-reduction} proves the claim.
\end{proof}

\section{Conclusion and open questions}\label{sec:questions}

We place approximate polynomial satisfiability and integral
closure membership over polynomial subalgebras in the Counting Hierarchy. Our APS algorithm uses finite Noether normalization to restrict attention to finitely many points of the image closure. These points arise as limits of solutions to a one-parameter polynomial system, and their membership can be tested through coefficients of a multivariate resultant. For integral closure membership, a global annihilator criterion and an effective degree bound give a deterministic polynomial-time reduction from non-integrality to APS.

These results yield Counting Hierarchy upper bounds for tensor border rank, the orbit-closure problems specified in Corollary~\ref{cor:aps-applications}, integral closure membership for ideals, and finiteness of polynomial maps. They also give $\mathsf{FP}^{\CH}$ constructions of polynomial-size hitting sets for $\overline{\mathsf{VP}}$. The results apply to the coefficient fields and input representations discussed above.

The connection between APS and integral dependence also works in the reverse direction. For a fresh variable $z$, a tuple $(f_1,\ldots,f_m)$ has an approximate zero precisely when $z$ is not integral over $K[zf_1,\ldots,zf_m]$ (Lemma~\ref{lem:aps-to-nonintegrality}). Consequently, APS and non-integrality are equivalent under deterministic polynomial-time many-one reductions.

\paragraph{Exact versus approximate satisfiability.}
Can APS be reduced to Hilbert's Nullstellensatz by a polynomial-time many-one reduction? The reverse reduction is known~\cite{BSTW26}. Such a reduction would transfer further upper bounds for exact satisfiability to APS and, through the equivalence above, to non-integrality. More broadly, can $\APS$ be placed in the polynomial hierarchy?

\paragraph{From decision to certificates.}
Our integral closure algorithm decides whether a polynomial is integral without constructing a monic relation or a global
annihilator certificate. Can the Counting Hierarchy upper bound be extended to accessing the coefficients of such certificates? Since a degree-bounded certificate may have exponentially many coefficients, a natural formulation asks whether one can select a valid certificate and compute a specified coefficient bit within $\CH$.

\newpage

\section*{Acknowledgements}
AP and SA would like to thank the \emph{SiMul Research group @ CRAN} for hosting the \emph{Geometry of Neural Networks and Tensors in Nancy} workshop. They would also like to thank Alex Massarenti who they met there and a discussion with whom resulted in the idea about formalizing the intuition about ``approaching'' the approximate solution. HS would like to thank Cornelius Brand for pointing us to the reference \cite{LL89}, that was very helpful in completing the geometric picture. AH would like to thank Markus Bl\"aser for hosting him for an internship in his group. AP would also like to thank Raimund Seidel for hosting him during this work. Finally, we would like to thank Markus Bl\"aser for his guidance throughout.

\section*{Funding}
HS was supported by the European Union (ERC, CountHom, 101077083). Views and opinions expressed are however those of the author(s) only and do not necessarily reflect those of the European Union or the European Research Council Executive Agency. Neither the European Union nor the granting authority can be held responsible for them.

\section*{AI Disclosure}
OpenAI Codex assisted with organizing and drafting the manuscript from the authors' supplied notes, checking references, and reviewing the arguments. It contributed to the global certificate and degree-bound proofs, the reduction from non-integrality to APS in \cref{sec:integral}, and the circuit-input and dimension arguments. Its assistance also affected the exposition and verification of the APS in Counting Hierarchy construction.

\bibliographystyle{alpha}
\bibliography{arxiv-submission}
\end{document}